\documentclass[final]{l4dc2025}

\usepackage{graphicx}
\usepackage{cleveref}
\usepackage{mleftright}
\usepackage{cases}
\usepackage{empheq}
\usepackage{bm}
\usepackage{bbm}
\usepackage{mathtools}
\usepackage{nicefrac}
\usepackage{listings}
\usepackage[colorinlistoftodos]{todonotes}
\usepackage{caption}
\usepackage{subcaption}
\usepackage{booktabs}
\usepackage{xspace}
\usepackage{microtype}
\usepackage{siunitx}
\usepackage{algorithm2e}
\usepackage{tikz}

\usepackage{algorithm}

\SetCommentSty{mycommfont}
\SetKwComment{Comment}{$\triangleright$\ }{}

\newtheorem{assumption}{Assumption}
\newtheorem{problem}{Problem}

\crefname{equation}{Eq.}{Eqs.}
\crefname{pluralequation}{Eqs.}{Eqs.}

\crefname{algorithm}{Algorithm}{Algorithm}

\crefname{figure}{Fig.}{Figs.}
\crefname{pluralfigure}{Figs.}{Figs.}

\crefname{section}{Sect.}{Sects.}
\crefname{pluralsection}{Sects.}{Sects.}

\crefname{table}{Table}{Table}
\crefname{pluraltable}{Tables}{Tables}

\crefname{definition}{Def.}{Def.}
\crefname{pluraldefinition}{Defs.}{Defs.}

\crefname{theorem}{Theorem}{Theorems}
\crefname{pluraltheorem}{Theorems}{Theorems}

\crefname{lemma}{Lemma}{Lemmas}
\crefname{plurallemma}{Lemmas}{Lemmas}

\crefname{example}{Example}{Example}
\crefname{pluralexample}{Examples}{Examples}

\crefname{problem}{Problem}{Problem}
\crefname{pluralproblem}{Problems}{Problems}

\crefname{assumption}{Assumption}{Assumption}
\crefname{pluralassumption}{Assumptions}{Assumptions}

\crefname{remark}{Remark}{Remark}
\crefname{pluralremark}{Remarks}{Remarks}

\crefname{proposition}{Proposition}{Proposition}
\crefname{pluralproposition}{Propositions}{Propositions}

\crefname{corollary}{Corollary}{Corollary}
\crefname{pluralcorollary}{Corollaries}{Corollaries}

\crefname{appendix}{Appendix}{Appendices}
\crefname{pluralappendix}{Appendices}{Appendices}

\newcommand*{\NN}{\mathbb{N}}

\newcommand*{\RR}{\mathbb{R}}

\newcommand{\system}{\mathcal{S}}
\newcommand*{\policy}{\ensuremath{\pi}}
\newcommand*{\policySpace}{\ensuremath{\Pi}}
\newcommand{\satprob}{\mathrm{Pr}}
\newcommand{\Prob}{\mathbb{P}}

\newcommand*{\distr}[1]{\Delta(#1)}

\newcommand{\tuple}[1]{\ensuremath{( #1 )}}

\newcommand*{\States}{\ensuremath{S}}
\newcommand*{\Actions}{\ensuremath{Act}}
\newcommand*{\initState}{\ensuremath{s_I}}

\newcommand*{\transfuncImdp}{\ensuremath{\mathcal{P}}}
\newcommand*{\IMDP}{\ensuremath{\tuple{\States,\Actions,\initState,\transfuncImdp}}}
\newcommand*{\imdp}{{\ensuremath{\mathcal{M}_\Interval}}}

\newcommand*{\scheduler}{\ensuremath{\sigma}}
\newcommand*{\Scheduler}{\ensuremath{\mathfrak{S}}}
\newcommand*{\schedulerSpace}{\ensuremath{\Scheduler}}

\newcommand*{\Interval}{\ensuremath{\mathbb{I}}}
\newcommand*{\plow}{\ensuremath{\check{p}}}
\newcommand*{\pupp}{\ensuremath{\hat{p}}}

\newcommand*{\xGoal}{\ensuremath{X_G}}
\newcommand*{\xUnsafe}{\ensuremath{X_U}}

\newcommand*{\sGoal}{\ensuremath{\States_G}}
\newcommand*{\sUnsafe}{\ensuremath{\States_U}}

\newcommand*{\horizon}{\ensuremath{h}}

\DeclareMathOperator*{\argmax}{arg\,max}

\newcommand{\Pre}{\mathsf{Pre}}

\newcommand{\Bin}{B}

\newcommand*{\diag}[1]{\ensuremath{{\mathrm{diag}(#1)}}}

\title[Data-Driven Formal Policy Synthesis for Stochastic Systems]{Data-Driven Yet Formal Policy Synthesis for \\ Stochastic Nonlinear Dynamical Systems}
\usepackage{times}

\author{%
 \Name{Mahdi Nazeri}${}^{1, 2}$ \Email{mahdi.nazeri@cs.ox.ac.uk}
 \AND
 \Name{Thom Badings}${}^1$ \Email{thom.badings@cs.ox.ac.uk}
 \AND
 \Name{Sadegh Soudjani}${}^2$ \Email{sadegh@mpi-sws.org}
 \AND
 \Name{Alessandro Abate}${}^1$ \Email{alessandro.abate@cs.ox.ac.uk}\\
 \addr ${}^1$Department of Computer Science, University of Oxford, Oxford, United Kingdom%
 \\
 \addr ${}^2$Max Planck Institute for Software Systems, Kaiserslautern, Germany%
}

\usepackage{enumitem}   
\setlist[enumerate]{%
        itemsep=0em, 
        topsep=0mm plus 2mm}

\newif\ifappendix
\global\appendixtrue

\begin{document}
\maketitle%
\begin{abstract}%
The automated synthesis of control policies for stochastic dynamical systems presents significant challenges. A standard approach is to construct a finite-state abstraction of the continuous system, typically represented as a Markov decision process (MDP). However, generating abstractions is challenging when (1) the system's dynamics are nonlinear, and/or (2) we do not have complete knowledge of the dynamics. In this work, we introduce a novel data-driven abstraction technique for nonlinear Lipschitz continuous dynamical systems with additive stochastic noise that addresses both of these issues. As a key step, we use samples of the dynamics to learn the enabled actions and transition probabilities of the abstraction. We represent abstractions as MDPs with intervals of transition probabilities, known as interval MDPs (IMDPs). These abstractions enable the synthesis of policies for the concrete nonlinear system, with probably approximately correct (PAC) guarantees on the probability of satisfying a specified control objective. Our numerical experiments illustrate the effectiveness and robustness of our approach in achieving reliable control under uncertainty.
\end{abstract}%
\section{Introduction}
\label{sec:introduction}
Formal policy synthesis is an area of control theory focusing on designing controllers that provably meet specific requirements~\citep{belta2017formal}. One such requirement is the \emph{(stochastic) reach-avoid task}: Compute a (control) policy such that, with at least a specified probability, the system reaches a set of goal states while avoiding unsafe states~\citep{DBLP:conf/cav/FanMM018,DBLP:journals/automatica/SummersL10}.
The state-of-the-art in policy synthesis for stochastic systems is, arguably, to abstract the system into a finite-state model that appropriately captures its behaviour~\citep{LSAZ21,DBLP:journals/automatica/AbatePLS08,DBLP:books/daglib/0032856}.
However, conventional abstractions often rely on precise and explicit representations of the system's dynamics, which are unavailable in many cases.

Fuelled by increasing data availability and advances in machine learning, \emph{data-driven abstractions} have emerged as an alternative to conventional model-based abstractions~\citep{DBLP:conf/adhs/MakdesiGF21,DBLP:journals/csysl/CoppolaPM23,lavaei2023compositional,DBLP:journals/corr/abs-2206-08069,DBLP:journals/automatica/HashimotoSKUD22,DBLP:conf/cdc/DevonportSA21,DBLP:journals/corr/abs-2303-17618,DBLP:journals/csysl/PeruffoM23,schon2024data}.
By incorporating techniques from formal verification~\citep{BaierKatoen08}, temporal logic~\citep{DBLP:conf/focs/Pnueli77}, and reachability analysis~\citep{DBLP:journals/arcras/AlthoffFG21}, data-driven abstractions can be used to synthesise policies despite incomplete knowledge of the system dynamics.
However, with only a few recent exceptions~\citep{DBLP:conf/l4dc/GraciaBLL24,DBLP:journals/corr/abs-2412-11343, lavaei2022constructing, DBLP:conf/hybrid/JacksonLFL21}, these data-driven abstractions apply to nonstochastic systems only, thus leaving an important gap in the literature.

\newpage
In this paper, we study discrete-time dynamical systems whose dynamics are composed of a deterministic nonlinear term and an additive stochastic noise term.
Following a data-driven paradigm, we assume only (black-box) sampling access to the stochastic noise.
For the nonlinear term, we require sampling access plus partial knowledge in the form of knowing the Lipschitz constant.
Given such a system and a reach-avoid task, we focus on the following problem:
\emph{Compute a policy such that the reach-avoid task is satisfied with at least a specific threshold probability $\rho \in [0,1]$.}

We address this problem by abstracting the system into a finite-state Markov decision process (MDP)~\citep{DBLP:books/wi/Puterman94}.
Inspired by~\cite{DBLP:conf/aaai/BadingsRA023}, we define the abstract~actions via \emph{backward reachability computations} on the dynamical system.
However, the approach from~\cite{DBLP:conf/aaai/BadingsRA023} only applies to systems with linear dynamics and leads to overly conservative abstractions.
To overcome these limitations, we introduce two data-driven aspects in our approach:

\begin{enumerate}
    \item \textbf{Data-driven backward reachability analysis:}
    Performing backward reachability computations on nonlinear systems is generally challenging~\citep{DBLP:conf/hybrid/Mitchell07,DBLP:journals/corr/abs-2209-14076}.
    We develop a novel data-driven method to \emph{underapproximate} backward reachable sets based on forward simulations of the dynamical system.
    Our method only requires differentiability of the nonlinear dynamics and leads to sound underapproximations of backward reachable sets.
    \item \textbf{Data-driven probability intervals:}
    We use statistical techniques to compute \emph{probably approximately correct} (PAC) intervals of transition probabilities, which we capture in an \emph{interval MDP} (IMDP)~\citep{DBLP:journals/ai/GivanLD00,DBLP:journals/ior/NilimG05}.
    While~\cite{DBLP:conf/aaai/BadingsRA023} also uses sampling techniques (using the scenario approach~\citep{DBLP:journals/arc/CampiCG21,romao2022tac}), their intervals are very loose.
    We instead use the classical Clopper-Pearson confidence interval~\citep{clopper1934use}, yielding tighter intervals~\citep{DBLP:journals/corr/abs-2404-05424}.
\end{enumerate}

In summary, our main contribution is a novel data-driven IMDP abstraction technique for nonlinear stochastic systems with incomplete knowledge of the dynamics.
Due to the PAC guarantee on each individual probability interval of the IMDP, we can use our abstraction to synthesise policies with PAC reach-avoid guarantees.
We showcase our abstraction technique on multiple benchmarks.

\smallskip
\noindent\textbf{Related work.}
Control of nonlinear systems against temporal tasks is an active research area~\citep{khalil2002nonlinear,belta2017formal}.
Since computing optimal policies is generally infeasible~\citep{Bertsekas.Shreve78}, many model-based abstraction techniques have been developed, often representing abstractions as (I)MDPs~\citep{DBLP:journals/siamads/SoudjaniA13, DBLP:journals/tac/LahijanianAB15,FAUST15, Huijgevoort2023SySCoRe,mathiesen2024,DBLP:conf/hybrid/Delimpaltadakis23}.
Particularly related here is~\cite{DBLP:conf/l4dc/GraciaBLL24}, who generate data-driven abstractions of switched stochastic systems into robust MDPs, by estimating the (unknown) noise distribution as a Wasserstein ball.
However,~\cite{DBLP:conf/l4dc/GraciaBLL24} resorts to a model-based approach to abstract the deterministic part of the dynamics, while we use sampling instead.
Also closely related are~\cite{DBLP:journals/jair/BadingsRAPPSJ23,DBLP:conf/aaai/BadingsRA023}, who, however, require the (deterministic) dynamics to be linear and known.

Existing methods over/underapproximate backward reachable sets by level set functions~\citep{DBLP:conf/amcc/YinPAS19,DBLP:conf/eucc/StipanovicHT03}, approximating operators on, e.g., zonotopes~\cite{DBLP:journals/tcad/YangZJO22}, piecewise affine bounding of the dynamics~\citep{DBLP:journals/corr/abs-2209-14076}, and Hamilton-Jacobi reachability analysis~\citep{DBLP:conf/cdc/BansalCHT17}.
However, most approaches are computationally expensive and are model-based, whereas we focus on data-driven techniques to obtain sound underapproximations.

Apart from abstraction, others recently studied control of stochastic systems using Lyapunov-like functions learned represented as neural networks~\citep{DBLP:journals/csysl/MathiesenCL23,DBLP:conf/cav/AbateGR24,DBLP:conf/aaai/ZikelicLHC23}, or using robust and scenario optimization~\citep{DBLP:journals/automatica/SalamatiLSZ24,DBLP:journals/tac/NejatiLJSZ23}.

\section{Problem Formulation}
\label{sec:problem}

A probability space $\tuple{\Omega, \mathcal{F}, \Prob}$ consists of an uncertainty space $\Omega$, a $\sigma$-algebra $\mathcal{F}$, and a probability measure $\Prob \colon \mathcal{F} \to [0,1]$.
A random variable $z$ is a measurable function $z \colon \Omega \to \RR^n$ for some $n \in \NN$.
The set of all distributions for a (continuous or discrete) set $X$ is $\distr{X}$.
The Cartesian product of an interval is $[a,b]^n$, for $a \leq b$, $n \in \NN$.
The element-wise absolute value of $x \in \RR^n$ is written as $|x|$.

\smallskip
\noindent\textbf{Stochastic systems.\,}
Consider a discrete-time nonlinear system $\system$ with additive stochastic noise:%
\begin{equation}
    \system : \,\, x_{k+1} = f(x_{k}, u_{k}) + w_k, \quad x_0 = x_I,
    \label{eq:DTSS}
\end{equation}
where $x_k \in \RR^n$ and $u_k \in \mathcal{U} \subset \RR^p$ are the state and control input at discrete time step $k \in \NN$, where $\mathcal{U} \subset \RR^p$ is compact.
The (deterministic) dynamics function  $f \colon \RR^n \times \mathcal{U} \to \RR^n$ is also called \emph{nominal dynamics}, and $x_I \in \RR^n$ is the initial state.
Moreover, $w_0, w_1, \ldots$ is a sequence of independent and identically distributed (i.i.d.) random variables, defined on the same probability space $(\Omega,\mathcal{F},\Prob)$.\footnote{For brevity, we assume $\Omega$ is a subset of $\RR^n$ and directly write $w \in \Omega$ to say that the random variable takes a value.}

\begin{assumption}
    \label{assumptions}
    The transition function $f$ is differentiable with bounded first-order partial derivatives, and 
    the measure $\Prob$ is absolutely continuous w.r.t. the Lebesgue measure.
    However, $\Prob$ itself is unknown.
\end{assumption}

Thus, while $f$ and $\Prob$ can be unknown, our method requires: (1)~the noise $w_k$ being additive, (2)~knowledge of, e.g., the Lipschitz constant of $f$, and (3)~independent sampling access to $f$ and $w_k$.

The inputs $u_k \in \RR^n$ are chosen by a \emph{(Markovian) policy} $\policy \coloneqq (\policy_0, \policy_1, \policy_2, \ldots)$, where each $\policy_k \colon \RR^n \to \mathcal{U}$, $k \in \NN$, is a measurable map from states to inputs.
We denote the set of all policies by~$\policySpace^\system$.
Fixing a policy $\policy$ defines a Markov process in the probability space of all trajectories~\citep{Bertsekas.Shreve78,DBLP:books/wi/Puterman94}, whose probability measure we denote by $\Prob^\system_{\policy}$. 

Given a policy $\policy$, we are interested in the probability of reaching a \emph{goal set} $\xGoal \subseteq \RR^n$ within $\horizon \in \NN \cup \{\infty\}$ steps, while never reaching an \emph{unsafe set} $\xUnsafe \subseteq \RR^n$.\footnote{Formally, $\xGoal$ and $\xUnsafe$ must be Borel-measurable~\citep{Salamon16}, but we glance over measurability details here.}
We call the triple $(\xGoal,\xUnsafe,\horizon)$ a \emph{reach-avoid specification}.
The \emph{reach-avoid probability} $\satprob^\system_{\policy}(\xGoal, \xUnsafe,\horizon)$ for this specification is
\begin{equation}
\begin{split}
    \label{eq:reachavoid_prob}
    \satprob^\system_{\policy}(\xGoal, \xUnsafe,\horizon) \coloneqq
    \Prob^\system_{\policy} \big\{ 
    \exists k & \in \{0,\ldots,\horizon\} \, : \,
    x_k \in \xGoal 
    \, \wedge
    ( \forall k' \in \{0,\ldots,k\} : x_{k'} \notin \xUnsafe )
    \big\}.
\end{split}
\end{equation}

\noindent We now have all the ingredients to formalise the problem that we wish to solve:
\begin{problem}
    \label{prob:Problem1}
    Suppose we are given a dynamical system $\system$, a reach-avoid specification $(\xGoal,\xUnsafe,\horizon)$, and a threshold probability $\rho \geq 0$.
    Compute a policy $\policy \in \policySpace^\system$ such that $\satprob^\system_{\policy}(\xGoal, \xUnsafe, \horizon) \geq \rho$.
\end{problem}

\smallskip
\noindent\textbf{Interval MDPs.}
We will abstract system $\system$ into an MDP with intervals of transition probabilities, known as an interval MDP (IMDP).
For an introduction to IMDPs, we refer to~\cite{Suilen2025}.
\begin{definition}[IMDP]
\label{def:IMDP}
An \emph{interval MDP} (IMDP) $\imdp$ is a tuple $\imdp \coloneqq \IMDP$, where
$\States$ is a finite set of states, $\initState \in S$ is the initial state,
$\Actions$ is a finite set of actions, with $\Actions(s) \subseteq \Actions$ the actions enabled in state $s \in S$, and
$\transfuncImdp \colon \States \times \Actions \rightharpoonup 2^{\distr{\States}}$ is a transition function\footnote{To model that not all actions may be enabled in a state, the transition function $\transfuncImdp$ is a partial map, denoted by $\rightharpoonup$.} defined for all $s \in \States, a \in \Actions(s)$ as
$
    \transfuncImdp(s,a) = \big\{\mu \in \distr{\States} : 
    \forall s' \in \States, \,\, 
    \mu(s') \in [\plow(s,a,s'), \pupp(s,a,s')] \subset [0,1]
    \big\}.
$
\end{definition}

Without loss of generality, we assume that $\transfuncImdp(s,a) \neq \emptyset$ for all $s \in \States$, $a \in \Actions(s)$.
We also call $[\plow(s,a,s'), \pupp(s,a,s')] \subseteq [0,1]$ the \emph{probability interval} for transition $(s,a,s')$.
Actions in an IMDP are chosen by a (Markovian) \emph{scheduler}\footnote{For clarity, we use the word \emph{scheduler} for (finite) IMDPs, whereas we use \emph{policy} for (continuous) dynamical systems.} $\scheduler = (\scheduler_0, \scheduler_1, \ldots)$, where each $\scheduler_k \colon \States \to \Actions$ is defined such that $\scheduler_k(s) = a \implies a \in \Actions(a)$.
The set of all Markov schedulers for $\imdp$ is denoted by $\schedulerSpace^\imdp$.

\newpage
An IMDP can be interpreted as a game between a scheduler that chooses actions and an \emph{adversary} that fixes distributions $P(s,a) \in \transfuncImdp(s,a)$ for all $s \in \States$, $a \in \Actions(s)$.
We assume a different probability can be chosen every time the same pair $(s,a)$ is encountered (called the \emph{dynamic} uncertainty model~\citep{DBLP:journals/mor/Iyengar05}).
We overload notation and write $P \in \transfuncImdp$ for fixing an adversary.

Fixing $\scheduler \in \schedulerSpace^\imdp$ and $P \in \transfuncImdp$ for $\imdp$ yields a Markov chain with (standard) probability measure $\Prob^\imdp_{\scheduler,P}$~\citep{BaierKatoen08}.
A reach-avoid specification for $\imdp$ is a tuple $(\sGoal,\sUnsafe,\horizon)$ of goal and unsafe states $\sGoal, \sUnsafe \subseteq \States$ and a horizon $\horizon \in \NN \cup \{\infty\}$ (we use this notation later in~\cref{subsec:IMDP_definition}).
The probability of satisfying this specification is written as $\satprob^\imdp_{\scheduler,P}(\sGoal, \sUnsafe,\horizon)$, and is defined based on $\Prob^\imdp_{\scheduler,P}$ analogously to \cref{eq:reachavoid_prob}.
An optimal (robust) scheduler $\scheduler^\star \in \schedulerSpace^\imdp$ is defined as
\begin{equation}
    \label{eq:optimal_policy}
    \scheduler^\star \in \argmax\nolimits_{\scheduler \in \schedulerSpace^\imdp} \, \min\nolimits_{P \in \transfuncImdp} \satprob^\imdp_{\scheduler,P}(\sGoal, \sUnsafe,\horizon).
\end{equation}
In practice, $\scheduler^\star$ can be computed using, e.g., robust value iteration~\citep{DBLP:conf/cdc/WolffTM12,DBLP:journals/mor/Iyengar05}.

\section{Finite-State IMDP Abstraction}
\label{sec:abstraction}

We present an abstraction of the system $\system$ into a finite IMDP.
Our approach extends~\cite{DBLP:conf/aaai/BadingsRA023} from \emph{linear} to \emph{nonlinear} systems, which has fundamental consequences for the practical computability of the abstraction.
For clarity, we first define the IMDP's states, actions, and transition function in this section, while we defer our novel contributions to compute these to \cref{sec:underapproximation,sec:probability_intervals}.

\smallskip
\noindent\textbf{Partition.\,}
Let $\mathcal{X} \subset \RR^n$ be a compact subset of the state space we want to capture by the abstraction.
We create a partition\footnote{The sets $\{R_1,\ldots,R_v\}$ form a partition of $\mathcal{X}$ if their union covers $\mathcal{X}$ and the interiors of all elements $R_i$ are disjoint.} of $\mathcal{X}$ into $v \in \NN$ convex polytopes $\{R_1, R_2, \ldots, R_v\}$, such that each region is defined as $R_i = \{x \in \RR^n : M_i x \leq b_i \}$, where $M_i \in \RR^{\xi_i \times n}$, $b_i \in \RR^{xi_i}$,~$\xi_i \in \NN$.
We append one element $R_\star = \RR^n \setminus \mathcal{X}$ to the partition, called the \emph{absorbing region}, which represents all states outside of $\mathcal{X}$.
Thus, the collection $\{R_1, R_2, \ldots, R_v\} \cup \{ R_\star \}$ covers the entire state space.

\begin{definition}[Scaled polytope]
    \label{def:scaled_polytope}
    Let $d_i \in R_i$ be the centre\footnote{In fact, we can choose any $d_i \in R_i$, but the centre is often convenient in practice. %
    } point of the region $R_i$.
    The convex polytope $R_i(\lambda) \subset \RR^n$ is defined as the version of $R_i$ scaled around $d_i$ by a factor of $\lambda \geq 0$, i.e.,
    $
    R_i(\lambda) = \left\{
    x \in \RR^n \mid M_ix \leq \lambda(b_i - M_i d_i) + M_i d_i
    \right\}.
    $
\end{definition}

Thus, for $\lambda=1$, the scaled region is the same, i.e., $R_i(1,d) = R_i$. 
Similarly, for any $\lambda < 1$, we have $R_i(\lambda,d_i) \subseteq R_i$.
We will use scaled regions later to define the abstract actions.

\subsection{Abstract IMDP definition}
\label{subsec:IMDP_definition}
\noindent\textbf{States.\,}
We define one IMDP state $s_i$ for each element $R_i$, plus one \emph{absorbing state} for $R_\star$, resulting in $S \coloneqq \{ s_1,\ldots,s_v, s_\star \}$.
An \emph{abstraction function} maps between continuous and~abstract~states:

\begin{definition}[Abstraction function]
    \label{def:abstraction_function}
    The \emph{abstraction function} $\mathcal{T} \colon \RR^n \to \States$ is defined as $\mathcal{T}(x) = s_i$ if $x \in R_i$.\footnote{If $x$ is on the boundary of two regions $R_i$, $R_j$, we arbitrarily choose $\mathcal{T}(x) = s_i$ or $\mathcal{T}(x) = s_i$. However, \cref{assumptions} implies that this occurs with probability zero, so this arbitrary choice does not affect the correctness of our algorithm.}
    We also define the preimage of $s_i$ under $\mathcal{T}$ as $\mathcal{T}^{-1}(s_i) = R_i$ for all $i=1,\ldots,v$.
\end{definition}

In other words, the IMDP state $s_i$ represents all continuous states $x \in R_i$.
The initial IMDP state is then defined as $\initState \coloneqq \mathcal{T}(x_I)$.
We map the reach-avoid specification $(\xGoal,\xUnsafe,\horizon)$ to the abstract IMDP by under- and overapproximating the goal and unsafe states, respectively, as
$\sGoal \coloneqq \{ s \in \States : \forall x \in \mathcal{T}^{-1}(s), \,\, x \in \xGoal \}$, and
$\sUnsafe \coloneqq \{ s \in \States : \exists x \in \mathcal{T}^{-1}(s), \,\, x \in \xUnsafe \}$.

\newpage
\noindent\textbf{Actions.\,}
Each IMDP action does not represent a single input $u_k \in \mathcal{U}$ (as is typically done in abstraction) but \emph{a collection} of inputs leading to a \emph{common state $x_{k+1}$}.
Without loss of generality, we define one action for each IMDP state (except $s_\star$), such that $\Actions \coloneqq \{a_1, \ldots, a_v\}$.
We then associate every pair $(s_i, a_j) \in \States \times \Actions$ with a so-called \emph{target set} $R_j(\lambda_{i \to j})$ that represents region $R_j$ scaled by a factor $\lambda_{i \to j} \geq 0$; see \definitionref{def:scaled_polytope}.
We discuss how we compute each factor $\lambda_{i \to j}$ in \cref{sec:underapproximation}.

Suppose the system's state is $x \in \RR^n$, associated with IMDP state $\mathcal{T}(x) = s_i$.
Choosing action $a_j \in \Actions$ in $s_i$ corresponds to choosing an input $u \in \mathcal{U}$ such that $f(x,u) \in R_j(\lambda_{i \to j})$.
Thus, this IMDP action $a_j$ must only be enabled in the state $s_i$ if for all $x \in \mathcal{T}^{-1}(s_i)$, there exists such an input $u \in \mathcal{U}$ for which $f(x,u)$ is contained in the target set $R_j(\lambda_{i \to j})$ of the state-action pair $(s_i, a_j)$.
To formalise this requirement, let $\Pre(X)$ be the \emph{backward reachable set} for a set $X \subset \RR^n$:
\begin{equation}
    \Pre(X) = %
    \left\{x' \in \RR^n \mid \exists u \in \mathcal{U} : f(x', u) \in X \right\}.
    \label{eq:BRS_set}
\end{equation}
Then, for every state $s_i \in \States$, the set of enabled actions $\Actions(s_i)$ is defined as
\begin{equation}
\begin{split}
    \label{eq:enabled_actions}
    \Actions(s_i) = \{
    a_j \in \Actions : 
    \exists \lambda_{i\to j} \in [0,\Lambda], \,\, \mathcal{T}^{-1}(s_i) \subseteq \Pre(R_j(\lambda_{i\to j}))
    \},
\end{split}
\end{equation}
where $\Lambda \in \RR_{\geq 0}$ is a global hyperparameter that prevents $R_j(\lambda_{i \to j})$ from becoming too large.

Computing $\Pre(\cdot)$ exactly is challenging for nonlinear dynamics~\citep{DBLP:conf/hybrid/Mitchell07,DBLP:journals/corr/abs-2209-14076}.
However, any \emph{underapproximation} preserves correction of our abstraction (albeit increasing conservatism).
In \cref{sec:underapproximation}, we present a data-driven method to compute such underapproximations.

\smallskip
\noindent\textbf{Transition function.\,}
Due to the stochastic noise in system $\system$, choosing  the IMDP action $a_j$ in IMDP state $s_i \in \States$ leads to the continuous successor state $f(x,u) + w \in \RR^n$, where $f(x,u) \in R_j(\lambda_{i\to j})$.
That is, for every possible $\hat{x} \in R_j(\lambda_{i\to j})$, we obtain a \emph{different probability distribution} over the continuous successor state $\hat{x} + w$.
Mathematically, let $\eta(\hat{x}, \bar{X}) \in [0,1]$ denote the probability that $\hat{x} + w$ is contained in a compact set $\bar{X} \subset \RR^n$, i.e.
$
    \eta(\hat{x}, \bar{X}) = \Prob \left\{ w \in \Omega : \hat{x} + w \in \bar{X} \right\}.
$
Then, the IMDP transition function $\transfuncImdp$ is defined for all $(s_i,a_j)$ by taking the min/max over $\eta$ as follows:
{\small
\begin{align}
    \transfuncImdp(s_i,a_j) = \Big\{\mu \in \distr{\States} : 
    \forall s' \in \States, \,\, & 
    \min_{\hat{x} \in R_j(\lambda_{i\to j})} \!\! \eta \big( \hat{x}, \mathcal{T}^{-1}(s') \big) 
    \leq \mu(s') 
    \leq 
    \max_{\hat{x} \in R_j(\lambda_{i\to j})} \!\! \eta \big( \hat{x}, \mathcal{T}^{-1}(s') \big)
    \Big\}.
    \label{eq:transition_function}
\end{align}
}
In \cref{sec:probability_intervals}, we will compute these bounds using samples of the noise $w \in \Omega$.

\subsection{Correctness of the abstraction}
In \cref{subsec:IMDP_definition}, we defined an IMDP abstraction $\imdp = \IMDP$ of system $\system$.
As is common in abstraction-based control, any IMDP scheduler $\scheduler \in \schedulerSpace^\imdp$ can be \emph{refined} into a policy $\policy \in \policySpace^\system$ for system $\system$.
Crucially, the reach-avoid probability for $\scheduler$ on $\imdp$ is a \emph{lower bound} on that for $\policy$ on $\system$:

\begin{theorem}[Policy synthesis~\citep{DBLP:journals/jair/BadingsRAPPSJ23}]
    \label{thm:existence_controller_main}
    Let $\imdp$ be the IMDP abstraction for dynamical system $\system$.
    For every IMDP scheduler $\scheduler \in \schedulerSpace^\imdp$, there exists a policy $\policy \in \policySpace^\system$ for $\system$~such~that
    \begin{equation}
        \label{eq:existence_controller_main}
        \min_{P \in \transfuncImdp} \satprob^\imdp_{\scheduler,P}(\sGoal, \sUnsafe,\horizon)
        \, \leq \,
        \satprob^\system_{\policy}(\xGoal, \xUnsafe,\horizon).
    \end{equation}
\end{theorem} 

As we discuss 
\ifappendix
    in~\cref{app:correctness},
\else
    in~\citet[App.~A]{L4DC_with_appendix},
\fi
\theoremref{thm:existence_controller_main} is based on a probabilistic extension of an \emph{alternating simulation relation}~\citep{DBLP:conf/concur/AlurHKV98},
The policy $\policy \in \policySpace^\system$ for which \theoremref{thm:existence_controller_main} holds can be derived recursively, by choosing inputs at every step $k \in \NN$ such that this relation is preserved.
Concretely, the \emph{refined policy} $\policy = (\policy_0, \policy_1, \ldots)$ for system $\system$ is defined for all $x \in \RR^n$ and $k \in \NN$ as
\begin{equation}
    \label{eq:refined_controller}
    \policy_k(x) \in \left\{ u \in \mathcal{U} : f(x,u) \in R_j(\lambda_{i \to j}) \right\},
\end{equation}
where $s_i = \mathcal{T}(x)$ and $a_j = \scheduler_k(s_i)$ are the current IMDP state and action.
In \cref{sec:underapproximation}, we discuss how we obtain $\policy_k(x)$ directly from the data-driven underapproximation of $\Pre(R_j(\lambda_{i \to j}))$.

\section{Data-Driven Underapproximations of Backward Reachable Sets}
\label{sec:underapproximation}

In this section, we compute the enabled actions $\Actions(s_i) \subseteq \Actions$ in each IMDP state $s_i \in \States$.
Recall from \cref{eq:enabled_actions} that action $a_j \in \Actions$ is enabled in state $s_i \in \States$ if $\mathcal{T}^{-1}(s_i) \subseteq \Pre(R_j(\lambda_{i \to j}))$.
As a key contribution, we present a data-driven method to compute the scaling factor $\lambda_{i \to j}$ and an underapproximation of $\Pre(R_j(\lambda_{i \to j}))$ for all $s_i \in \States$ and $a_j \in \Actions(s_i)$.

\smallskip
\noindent\textbf{Data collection.}
The core idea is to underapproximate each set $\Pre(R_j(\lambda_{i \to j}))$, based on forward simulations of the nominal dynamics function $\hat{x}_{\ell} = f(x_\ell, u_\ell)$.
Since we assumed sampling access to $f$, we can easily obtain such a set of samples.
Let us denote the resulting set of $K \in \NN$ samples by
\begin{equation}
    \label{eq:dataset}
    \mathcal{D}_K = \left\{ (x_\ell,u_\ell,f(x_\ell,u_\ell)) : \ell = 1,\ldots,K, \, x_\ell \in \mathcal{X}, \, u_\ell \in \mathcal{U} \right\}.
\end{equation}
Without loss of generality, we assume to obtain these samples by a uniform gridding of $\mathcal{X}$~and~$\mathcal{U}$.
While more sophisticated approaches may lead to better results, we leave this for future work.

We describe how we use the dataset $\mathcal{D}_K$ to underapproximate $\Pre(R_j(\lambda_{i \to j}))$ for a fixed $s_i \in \States$ and $a_j \in \Actions$.
We repeat this procedure for all other state-action pairs to compute all enabled actions.

\smallskip
\noindent\textbf{Fixing $\lambda_{i \to j}$ upfront.}
Consider a fixed value for $\lambda_{i \to j} > 0$, which fixes the target set $R_j(\lambda_{i \to j})$ that defines the semantics of action $a_j$.
By definition, the $x$-component of every sample $(x,u,x') \in \mathcal{D}_K$ for which $x' \in R_j(\lambda_{i \to j})$ is contained in $\Pre(R_j(\lambda_{i \to j}))$.
Moreover, due to the differentiability of the dynamics (see \cref{assumptions}), there exists a region $Y$ around $x$ such that, for all $y \in Y$, $f(y,u)$ is also contained in $R_j(\lambda_{i \to j})$.
Thus, this region $Y$ is also contained in $\Pre(R_j(\lambda_{i \to j}))$.

For a fixed input $\hat{u}$, the Jacobian of $f(x, \hat{u})$ is the matrix $J \in \RR^{n \times n}$, whose entries are defined as $J_{pq} = \frac{\partial f(x, \hat{u})_p}{\partial x_q}$, $p,q = 1,\ldots,n$.
We define $J^+(R_i) \in \RR^{n \times n}$ as the matrix whose entries $J^+(R_i)_{pq} \in \RR_{\geq 0}$ are defined as the supremum over the absolute value of $J_{pq}(x)$ for all $x \in R_i$, i.e.,
$
    J^+(R_i)_{pq} = \sup \left\{ |J_{pq}(x)| : x \in R_i \right\}.
$
We use the matrix $J^+(R_i)$ to derive the following theorem.
\begin{theorem}
    \label{thm:Jacobian}
    For all $x_1, x_2 \in R_i$ it holds that
    $
        |f(x_1, u_\ell) - f(x_2, u_\ell)| \leq J^+(R_i) \cdot |x_1 - x_2|.
    $
\end{theorem}

\begin{figure}[t!]
\centering
\begin{tikzpicture}[scale=1, font=\footnotesize]
    \draw[step=1cm,gray,dashed] (0.8,0.8) grid (4.2,3.2);

    \fill[red!10] (1,1) rectangle (2,2) node [pos=0.5, yshift=0.4cm, xshift=-0.3cm, above, black] {$R_i$};

    \fill[orange!10] (3,2) rectangle (4,3) node [pos=0.5, yshift=0.5cm, above, orange!80!black] {$R_j(0.95)$};
    \fill[orange!80] (3.05,2.05) rectangle (3.95,2.95);

    \node (x1) [circle, fill=black, inner sep=1pt, minimum size=1pt]at (1.3,1.6) {};
    \node (x1plus) [circle, fill=black, inner sep=1pt, minimum size=1pt] at (3.3,2.7) {};
    \draw[->, thin, bend left] (x1) to (x1plus);
    \draw[fill, black, opacity=0.1] (1.05,1.35) rectangle (1.55,1.85);
    \draw[fill, black, opacity=0.1] (3.05,2.45) rectangle (3.55,2.95);
    
    \node (x2) [circle, fill=black, inner sep=1pt, minimum size=1pt]at (1.7,1.8) {};
    \node (x2plus) [circle, fill=black, inner sep=1pt, minimum size=1pt] at (3.65,2.6) {};
    \draw[->, thin, bend left] (x2) to (x2plus);
    \draw[fill, black, opacity=0.1] (1.4,1.5) rectangle (2.0,2.1);
    \draw[fill, black, opacity=0.1] (3.35,2.3) rectangle (3.95,2.90);

    \node (x3) [circle, fill=black, inner sep=1pt, minimum size=1pt]at (1.6,1.3) {};
    \node (x3plus) [circle, fill=black, inner sep=1pt, minimum size=1pt] at (3.55,2.4) {};
    \draw[->, thin, bend left] (x3) to (x3plus);
    \draw[fill, black, opacity=0.1] (1.25,0.95) rectangle (1.95,1.65) node [right, pos=0.5, opacity=1, xshift=0.2cm] {$A_j(x)$};
    \draw[fill, black, opacity=0.1] (3.2,2.05) rectangle (3.9,2.75) node [below, pos=0.5, opacity=1, yshift=-0.24cm] {$B^\infty_{r(\cdot)}(x')$};
    
    \draw[thick,->] (0.6,0.6) -- (4.2,0.6) node[right] {$x^{(1)}$};
    \draw[thick,->] (0.6,0.6) -- (0.6,3.2) node[left] {$x^{(2)}$};
\end{tikzpicture}%
\begin{tikzpicture}[scale=1, font=\footnotesize]
    \draw[step=1cm,gray,dashed] (0.8,0.8) grid (4.2,3.2);

    \fill[red!10] (1,1) rectangle (2,2) node [pos=0.5, yshift=0.4cm, xshift=-0.3cm, above, black] {$R_i$};

    \fill[orange!10] (3,2) rectangle (4,3) node [pos=0.5, yshift=0.5cm, above, orange!80!black] {$R_j(0.80)$};
    \fill[orange!80] (3.2,2.2) rectangle (3.8,2.8);

    \node (x1) [circle, fill=black, inner sep=1pt, minimum size=1pt]at (1.3,1.6) {};
    \node (x1plus) [circle, fill=black, inner sep=1pt, minimum size=1pt] at (3.3,2.7) {};
    \draw[->, thin, bend left] (x1) to (x1plus);
    \draw[fill, black, opacity=0.1] (1.2,1.5) rectangle (1.4,1.7);
    \draw[fill, black, opacity=0.1] (3.2,2.6) rectangle (3.4,2.8);
    
    \node (x2) [circle, fill=black, inner sep=1pt, minimum size=1pt]at (1.7,1.8) {};
    \node (x2plus) [circle, fill=black, inner sep=1pt, minimum size=1pt] at (3.65,2.6) {};
    \draw[->, thin, bend left] (x2) to (x2plus);
    \draw[fill, black, opacity=0.1] (1.55,1.65) rectangle (1.85,1.95);
    \draw[fill, black, opacity=0.1] (3.5,2.45) rectangle (3.8,2.75);

    \node (x3) [circle, fill=black, inner sep=1pt, minimum size=1pt]at (1.6,1.3) {};
    \node (x3plus) [circle, fill=black, inner sep=1pt, minimum size=1pt] at (3.55,2.4) {};
    \draw[->, thin, bend left] (x3) to (x3plus);
    \draw[fill, black, opacity=0.1] (1.4,1.1) rectangle (1.8,1.5) node [right, pos=0.5, opacity=1, xshift=0.1cm] {$A_j(x)$};
    \draw[fill, black, opacity=0.1] (3.35,2.2) rectangle (3.75,2.6) node [below, pos=0.5, opacity=1, yshift=-0.13cm] {$B^\infty_{r(\cdot)}(x')$};
    
    \draw[thick,->] (0.6,0.6) -- (4.2,0.6) node[right] {$x^{(1)}$};
    \draw[thick,->] (0.6,0.6) -- (0.6,3.2) node[left] {$x^{(2)}$};
\end{tikzpicture}%
\begin{tikzpicture}[scale=1.733, font=\footnotesize]
    \draw[step=1cm,gray,dashed] (0.8,0.8) grid (2.2,2.2);

    \fill[red!10] (1,1) rectangle (2,2) node [pos=0.5, yshift=0.8cm, xshift=-0.6cm, above, black] {$R_i$};

    \draw[step=.1cm,gray,densely dotted] (1,1) grid (2,2);

    \draw[fill, orange!80, draw=black] (1.1,1.8) rectangle (1.2,1.9) node[pos=0.5, right, black] {$\phi$};

    \node (x1) [circle, fill=black, inner sep=1pt, minimum size=1pt] at (1.3,1.6) {};
    \node (x3) [circle, fill=black, inner sep=1pt, minimum size=1pt] at (1.7,1.4) {};

    \draw[fill, black, opacity=0.1] (1.0,1.3) rectangle (1.6,1.9) node [right, pos=0.5, opacity=1, xshift=0cm] {$x_1$} node [right, pos=0.0, opacity=1, xshift=-0.1cm, yshift=0.2cm] {$A_j(x_1)$};
    \draw[fill, black, opacity=0.1] (1.1,0.8) rectangle (2.3,2.0) node [right, pos=0.5, opacity=1, xshift=0cm] {$x_2$} node [right, pos=0.0, opacity=1, xshift=-0.1cm, yshift=0.2cm] {$A_j(x_2)$};

    \draw[thick,->] (0.7,0.7) -- (2.2,0.7) node[right] {$x^{(1)}$};
    \draw[thick,->] (0.7,0.7) -- (0.7,2.2) node[left] {$x^{(2)}$};
\end{tikzpicture}%
%
\vspace{-3em}%
\caption{
Three samples $(x,u,x')$ with $x' \in R_j(\lambda_{i \to j})$, and the balls $B_{r(\cdot)}^\infty(x')$ around each $x'$ and $A_j(x)$ around each $x$ (shown in gray) for fixed values of $\lambda_{i \to j} = 0.95$ (left) and $\lambda_{i \to j} = 0.80$ (middle). 
On the right, we show sets $A_j(x)$ for two samples such that voxel $\phi \in \Phi(R_i)$ is contained.
}
\label{fig:BRS}
\end{figure}

\noindent
We prove \theoremref{thm:Jacobian}
\ifappendix
    in~\cref{app:Jacobian}.
\else
    in~\citet[App.~B]{L4DC_with_appendix}.
\fi
Note that the maximum element of $J^+(R_i)$ is a local Lipschitz constant of $f$ with respect to changes in $x$, within the region $R_i$.
Thus, if computing the Jacobian $J$ is difficult, e.g., when $f$ is not known explicitly, we can use an (upper bound on the) Lipschitz constant instead.
We use \theoremref{thm:Jacobian} to underapproximate $\Pre(R_j(\lambda_{i\rightarrow j}))$ as follows:

\begin{definition}
    Let $(x,u,x') \in \mathcal{D}_K$ with $x \in R_i$, $x' \in R_j$. The radius $r_j(x', \lambda_{i \to j})$ of the largest $x'$-centered $L^\infty$-ball\footnote{The (open) $L^\infty$-ball $B^\infty_\epsilon(x')$ of size $\epsilon \geq 0$ centered at $x'$ is defined as $B^\infty_\epsilon(x) = \{y \in \RR^n : \lVert x'-y \rVert_\infty < \epsilon \}$.} contained in $R_j(\lambda_{i \to j})$ is $r_j(x', \lambda_{i \to j}) \coloneqq \max\{ \epsilon \geq 0 : B_\epsilon^\infty(x') \subseteq R_j(\lambda_{i \to j})\}$.    
\end{definition}

\begin{theorem}[Underapproximate backward reach. set]
    \label{thm:underapproximation}
    Fix $s_i \in \States$, $a_j \in \Actions$, {and $\lambda_{i \to j} \geq 0$.}
    Let $(x,u,x') \in \mathcal{D}_K$ be a sample with $x' \in R_j(\lambda_{i \to j})$ and define the set $A_j(x) \coloneqq \big\{ y \in R_i : {\big\lVert J^+(R_i) \cdot | x - y | \big\rVert_\infty} \leq r_j(x', \lambda_{i \to j}) \big\}$.
    Then, it holds that $A_j(x) \subseteq \Pre(R_j(\lambda_{i\rightarrow j}))$.
\end{theorem}

By using \theoremref{thm:underapproximation} for multiple samples, we obtain $\cup_x A_j(x) \subseteq \Pre(R_j(\lambda_{i\rightarrow j}))$.
This idea is visualised in \figureref{fig:BRS}, showing $R_j(\lambda_{i \to j})$ in orange and three samples $(x,u,x')$.
The shaded squares around each $x'$ are the balls $B_{r_j(x',\lambda_{i \to j})}^\infty(x')$, and the squares around each $x$ are the sets $A_j(x)$ that form the underapproximation of the backward reachable set.
Observe that, for the higher value of $\lambda_{i \to j} = 0.95$, we obtain larger sets $A_j(x)$ and thus a larger underapproximation.
However, a higher $\lambda_{i \to j}$ also leads to a larger target set $R_j(\lambda_{i \to j})$, and thus to a more conservative abstraction.

\smallskip
\noindent\textbf{Algorithm with variable $\lambda_{i \to j}$.}
Fix a state $s_i$ and an action $a_j$.
To determine if $a_j$ is enabled in $s_i$, we need to check if the union of all sets $A_j(x)$ covers $\mathcal{T}^{-1}(s_i) = R_i$.
Moreover, the question remains what value of $\lambda_{i \to j}$ we should use in practice.
To this end, we propose an algorithm that chooses $\lambda_{i \to j}$ based on the samples $(x,u,x') \in \mathcal{D}_K$ available.
For brevity, define $\mathcal{D}_K(R_j) \subset \mathcal{D}_K$ as the subset of samples for which $x' \in R_j$, i.e.,
$
\mathcal{D}_K(R_j) = \{ (x,u,x') \in \mathcal{D}_K : x' \in \mathcal{T}^{-1}(s_j) \}.
$
As described in \cref{algo:BRS}, we compute the enabled actions $\Actions(s_i)$ for all $s_i \in \States$ as follows:
\begin{enumerate}
    \item As shown in \figureref{fig:BRS} (right), we create a uniform tiling of $R_i$ into hyperrectangles that we call \emph{voxels}.
    The set of all $m_i \in \NN$ voxels for $R_i$ is $\Phi(R_i) = \{ \phi_1, \ldots, \phi_{m_i} \}$, where each $\phi_\ell \subset \RR^n$ and $\cup_{\ell = 1}^{m_i} \phi_\ell = R_i$. Let $c_\phi, \delta_\phi \in \RR^n$ be the centre and radius of voxel $\phi \in \Phi(R_i)$, respectively. 
    \item Fix a voxel $\phi \in \Phi(R_i)$ and a sample $(x,u,x') \in \mathcal{D}_K(R_j)$. With overloading of notation, let $\lambda^\star(\phi, (x,u,x'))$ be the smallest value of $\lambda_{i \to j}$ such that $\phi$ is completely contained in $A_j(x)$, i.e., $\phi \subseteq A_j(x)$. In practice, we find an overapproximation of $\lambda^\star(\phi, (x,u,x'))$ defined as
    \begin{equation}
        \lambda^+(\phi, (x,u,x')) \coloneqq 1 +
        \frac{\big\lVert J^+(R_i) \cdot (| x - c_\phi| + \delta_\phi)\big\rVert_\infty - r_j(x', 1)}{r_j(x', 2) - r_j(x', 1)}  \geq \lambda^\star(\phi, (x,u,x')).
        \label{eq:voxel_distance}
    \end{equation}
    \item Then, we compute the actual value of $\lambda_{i \to j}$ as the maximum of $\lambda^+(\phi, (x,u,x'))$ over all $\phi \in \Phi(R_i)$ and the minimum over all $(x,u,x') \in \mathcal{D}_K(R_j)$:
    \begin{equation}
        \lambda_{i \to j} = \max_{\phi \in \Phi(R_i)} \, \min_{(x,u,x') \in \mathcal{D}_K(R_j)} \, \lambda^+(\phi, (x,u,x')).
        \label{eq:compute_lambda}
    \end{equation}
    As shown in \figureref{fig:BRS} (right), we thus find a factor $\lambda_{i \to j}$ such that \emph{every voxel $\phi$ is covered by a ball around some sampled state $x'$}.
    Since $\cup_{\phi \in \Phi(R_i)} = R_i$, it follows that $R_i \subseteq \Pre(R_j(\lambda_{i\rightarrow j}))$.
    \item We check whether $\lambda_{i \to j}$ satisfies the global upper bound $\Lambda$. If $\lambda_{i \to j} \leq \Lambda$, then we conclude that $a_j \in \Actions(s_i)$ for $\lambda_{i \to j}$.
    If, on the other hand, $\lambda_{i \to j} > \Lambda$, then we set $a_j \notin \Actions(s_i)$.
\end{enumerate}

\begin{remark}[Role of the maximum scaling factor $\Lambda$]
    The hyperparameter $\Lambda$ controls the number of actions enabled in the IMDP.
    A higher $\Lambda$ results in more enabled IMDP actions, resulting in a more accurate but larger abstraction.
    In the experiments, we explore the effect of $\Lambda$ in practice.
\end{remark}

\noindent\textbf{Policy refinement.}
Our approach leads directly to a strategy for obtaining the refined policy $\policy_k(x)$ in \cref{eq:refined_controller}.
Suppose that at time $k$, the continuous state $x_k \in R_i$ corresponds with IMDP state $s_i = \mathcal{T}(x_k)$, and suppose the optimal IMDP action is $a_j = \scheduler(s_i)$.
Let $\phi \in \Phi(R_i)$ be the voxel containing $x_k$.
Then, we choose $\policy_k(x_k)$ as the input $u \in \mathcal{U}$ that attains the minimal $\lambda^+(\phi,(x,u,x'))$ over all samples $(x,u,x') \in \mathcal{D}_K(R_j)$ in \cref{eq:compute_lambda}.
As required, this leads to $f(x_k,u) \in R_j(\lambda_{i \to j})$ by construction.
Thus, we obtain a refined policy that is constant within each voxel.

\begin{algorithm}[t]
{\small
\DontPrintSemicolon
\caption{Computing enabled actions by underapproximating backward reachable sets}
\label{algo:BRS}
\KwData{Samples $\mathcal{D}_K = \{(x_\ell, u_\ell, \hat x_{\ell}=f(x_\ell, u_\ell))\}_{\ell=1}^{K}$; max. scaling factor $\Lambda > 0$}
\KwResult{Enabled actions $\Actions(s_i) \subseteq \Actions$ for all IMDP states $s_i \in \States$}
\hrule
\,$\Actions(s_i) \gets \emptyset \,\, \forall s_i \in \States$\Comment*[r]{Initialise enabled actions}
\,\For{$j = 1,\ldots,v$}{
    \,$\mathcal{D}_K(R_j) \gets \{ (x,u,x') \in \mathcal{D}_K : x' \in \mathcal{T}^{-1}(s_j) \}$\Comment*[r]{Find samples leading to successor in $s_j$}
    \,\For{$i = 1,\ldots,v$}{
        \,$\Phi(R_i) \gets \{\phi_1,\ldots,\phi_{m_i}\}$\Comment*[r]{Define voxelised representation of $R_i$}
        \,\For{$\phi \in \Phi(R_i)$}{
            \,\For{$(x,u,x') \in \mathcal{D}_K(R_j)$}{
                \,\If{$x \in R_i$}{
                    \, $\lambda^+(\phi, (x,u,x')) \gets  1+ \frac{\big\lVert J^+(R_i) \cdot (| x - c_\phi| + \delta_\phi)\big\rVert_\infty - r_j(x', 1)}{r_j(x', 2) - r_j(x', 1)} $
                    }
                }
            }
            \,$\lambda_{i \to j} \gets \max_{\phi \in   \Phi(R_i)} \, \min_{(x,u,x') \in \mathcal{D}_K(R_j)} \, \lambda^+(\phi, (x,u,x'))$\Comment*[r]{Compute scaling factor}
            \If{$\lambda_{i \to j} \leq \Lambda$}
            {
                \,$\Actions(s_i) \gets \Actions(s_i) \cup \{ a_j \}$\Comment*[r]{Enable action if $\lambda_{i \to j}$ is below max. scaling factor $\Lambda$}
            }
        }
    }
}
\end{algorithm}

\section{Computing Probability Intervals with Data} 
\label{sec:probability_intervals}

We compute bounds on the probability intervals in \cref{eq:transition_function} by sampling the noise $w \in \Omega$.
We focus on a fixed transition $(s_i, a_j, s')$ and repeat the procedure for all state-action pairs.
First, observe that\footnote{We write $R_j(\lambda_{i\to j}) + w = \{\alpha + w : \alpha \in R_j(\lambda_{i\to j})\}$ for the Minkowski sum between $R_j(\lambda_{i\to j})$ and $w$.}
\begin{align}
    \label{eq:probability_lb}
    \check{P}_{i,j}(s') &\coloneqq \Prob\left\{ w \in \Omega : R_j(\lambda_{i\to j}) + w \subseteq \mathcal{T}^{-1}(s') \right\}
    \leq 
    \min_{\hat{x} \in R_j(\lambda_{i\to j})} \!\! \eta \big( \hat{x}, \mathcal{T}^{-1}(s') \big),
    \\
    \label{eq:probability_ub}
    \hat{P}_{i,j}(s') &\coloneqq \Prob\left\{ w \in \Omega : R_j(\lambda_{i\to j}) + w \cap \mathcal{T}^{-1}(s') \neq \emptyset \right\}
    \geq
    \max_{\hat{x} \in R_j(\lambda_{i\to j})} \!\! \eta \big( \hat{x}, \mathcal{T}^{-1}(s') \big),
\end{align}
where the probabilities $\check{P}_{i,j}(s')$, $\hat{P}_{i,j}(s')$ are the (unknown) success probabilities of a Bernoulli random variable.
Thus, fixing a set $\{w^{(1)},\ldots,w^{(N)}\} \in \Omega^N$ of $N \in \NN$ noise samples induces the binomial distributions $\Bin(N,\check{P}_{i,j}(s'))$ and $\Bin(N,\hat{P}_{i,j}(s'))$.\footnote{We write $\Bin(n,p)$ to denote a binomial distribution with $n \in \NN$ experiments and success probability $p \in [0,1]$.}
Observing $\{w^{(1)},\ldots,w^{(N)}\} \in \Omega^N$ yields samples $\check{N}_{i,j}(s') \sim \Bin(N,\check{P}_{i,j}(s'))$ and $\hat{N}_{i,j}(s') \sim \Bin(N,\hat{P}_{i,j}(s'))$ from these binomials.
\ifappendix
    (see~\cref{app:probability_intervals} for an explicit definition).
\else
    (see~\citet[App.~D]{L4DC_with_appendix} for an explicit definition).
\fi
\cite{DBLP:conf/aaai/BadingsRA023,DBLP:conf/eucc/BadingsRA024} leverage the scenario approach~\citep{DBLP:journals/arc/CampiCG21} to estimate $\check{P}_{i,j}(s')$ and $\hat{P}_{i,j}(s')$ based on $\check{N}_{i,j}(s')$ and $\hat{N}_{i,j}(s')$ as intervals.
However, as recently pointed out by~\cite[Theorem~2]{DBLP:journals/corr/abs-2404-05424}, tighter intervals can be obtained by using the \emph{Clopper-Pearson interval}, a well-known statistical method for calculating binomial confidence intervals~\citep{clopper1934use,newcombe1998two}:

\begin{theorem}[Clopper-Pearson interval]
    \label{thm:PAC_interval}
    Let $\{w^{(1)},\ldots,w^{(N)}\} \in \Omega^N$, and let  $\beta \in (0,1)$. 
    For fixed $s_i, s' \in \States$ and $a_j \in \Actions(s_i)$, compute $\check{N}_{i,j}(s')$ and $\hat{N}_{i,j}(s')$.
    Then, it holds that
    \begin{equation}
        \label{eq:thm_interval}
        \Prob^N \Big\{ 
        \{w^{(1)},\ldots,w^{(N)}\} \in \Omega^N
        :
        \check{P}_\text{lb} \leq \transfuncImdp(s_i, a_j)(s') \leq \hat{P}_\text{ub}
        \Big\} \geq 1-\beta,
    \end{equation}
    where $\check{P}_\text{lb} = 0$ if $\check{N}_{i,j}(s') = 0$, and otherwise, $\check{P}_\text{lb}$ is the solution to
    \begin{equation}
        \label{eq:thm_lower_bound}
        \frac{\beta}{2} = \sum\nolimits_{i=\check{N}_{i,j}(s')}^{N} \binom Ni \cdot (\check{P}_\text{lb})^{i} \cdot (1-\check{P}_\text{lb})^{N-i},
    \end{equation}
    and $\hat{P}_\text{ub} = 1$ if $\check{N}_{i,j}(s') = N$, and otherwise, $\hat{P}_\text{ub}$ is the solution to
    \begin{equation}
        \label{eq:thm_upper_bound}
        \frac{\beta}{2} = \sum\nolimits_{i=0}^{\hat{N}_{i,j}(s')} \binom Ni \cdot (\hat{P}_\text{ub})^i \cdot (1-\hat{P}_\text{ub})^{N-i}.
    \end{equation}
\end{theorem}
\begin{proof}
    The proof follows by applying the Clopper-Pearson interval~\citep{clopper1934use,Thulin_2014} to the binomials $\check{N}_{i,j}(s') \sim \Bin(N,\check{P}_{i,j}(s'))$ and $\hat{N}_{i,j}(s') \sim \Bin(N,\hat{P}_{i,j}(s'))$,~which~yields
    \begin{equation}
        \label{eq:thm_interval:proof1}
        \Prob^N \big\{ \check{P}_\text{lb} \leq \check{P}_{i,j}(s') \big\} \geq 1-\nicefrac{\beta}{2},
        \quad \text{and} \quad
        \Prob^N \big\{ \hat{P}_{i,j}(s') \leq \hat{P}_\text{ub} \big\} \geq 1-\nicefrac{\beta}{2}.
    \end{equation}
    Combining \cref{eq:thm_interval:proof1} with \cref{eq:probability_lb,eq:probability_ub} through the union bound, we obtain \cref{eq:thm_interval}. 
\end{proof}
\theoremref{thm:PAC_interval} asserts that each interval is \emph{correct with probability $\geq 1-\beta$}.
Combining \theoremref{thm:existence_controller_main,thm:PAC_interval} leads to a statistical solution to \cref{prob:Problem1}; the proof is analogous to~\citet[Thm.~2]{DBLP:journals/jair/BadingsRAPPSJ23}:
\begin{corollary}
    \label{corollary:solution}
    Let $\imdp'$ be the IMDP abstraction with probability intervals obtained via \theoremref{thm:PAC_interval}.
    For every IMDP scheduler $\scheduler \in \schedulerSpace^{\imdp'}$, there exists a policy $\policy \in \policySpace^\system$ for system $\system$ such that
    \begin{equation*}
        \label{eq:overall_bound}
        \Prob\Big\{
        \min_{P \in \transfuncImdp} \satprob^{\imdp'}_{\scheduler,P}(\sGoal, \sUnsafe,\horizon)
        \leq
        \satprob^\system_{\policy}(\xGoal, \xUnsafe,\horizon)
        \Big\}
        \geq \max(0, 1-\beta \cdot |\States|^2 \cdot |\Actions|).
    \end{equation*}
\end{corollary}

The factor of $|\States|^2 \cdot |\Actions|$ in \cref{eq:overall_bound} comes from the maximum possible number of IMDP transitions, which is the number of $(s,a,s')$ triples such that $s \in \States$, $a \in \Actions(s)$, and $\transfuncImdp(s,a)(s') > 0$.

\corollaryref{corollary:solution} carries an important message: For any system $\system$ (that satisfies \cref{assumptions}), the IMDP abstraction $\imdp'$ with probability intervals given by \theoremref{thm:PAC_interval} leads, with at least probability $1-\beta \cdot |\States|^2 \cdot |\Actions|$, to a solution to \cref{prob:Problem1}.
In practice, we can again use \cref{eq:refined_controller} to refine any IMDP scheduler $\scheduler \in \schedulerSpace^{\imdp'}$ into the corresponding policy $\policy \in \policySpace^\system$ for system $\system$ that solves \cref{prob:Problem1}.

\section{Experimental Evaluation}
\label{sec:experiments}

We conduct experiments on (1) an inverted pendulum, (2) a harmonic oscillator with nonlinear damping, and (3) a car parking benchmark with nonlinear control.
Details on each benchmark are
\ifappendix
    in~\cref{app:Experiments}.
\else
    in~\citet[App.~E]{L4DC_with_appendix}.
\fi
We implement our approach in Python, using robust value iteration~\citep{DBLP:conf/cdc/WolffTM12,DBLP:journals/mor/Iyengar05}, implemented in the model checker PRISM~\citep{DBLP:conf/cav/KwiatkowskaNP11} to compute optimal schedulers as per \cref{eq:optimal_policy}.
All experiments ran parallelized on a computer with 32 3.3 GHz cores and 128 GB  of RAM. 
For all experiments, we use an error rate of $\beta=1-\frac{0.05}{T}$ on every IMDP transition, leading to an overall confidence probability (as per \corollaryref{corollary:solution}) of $0.95$ for an IMDP abstraction with $T$ transitions.

\begin{figure}[t!]
\floatconts
{fig:results_car}%
{\vspace{-1.5em}\caption{Reach-avoid probabilities $\satprob^\system_{\policy}(\xGoal, \xUnsafe,\horizon)$ for the car benchmark with (a) probability intervals from \theoremref{thm:PAC_interval} and (b) the approach from \cite{DBLP:conf/aaai/BadingsA00PS22}, both with $N = 10\,000$ samples. Fig.~(c) shows simulated trajectories under the resulting policy from \cref{eq:refined_controller} for our method.}}%
{%
\subfigure[Clopper-Pearson intervals]{%
\label{fig:results_car:CP}%
\includegraphics[height=4.2cm]{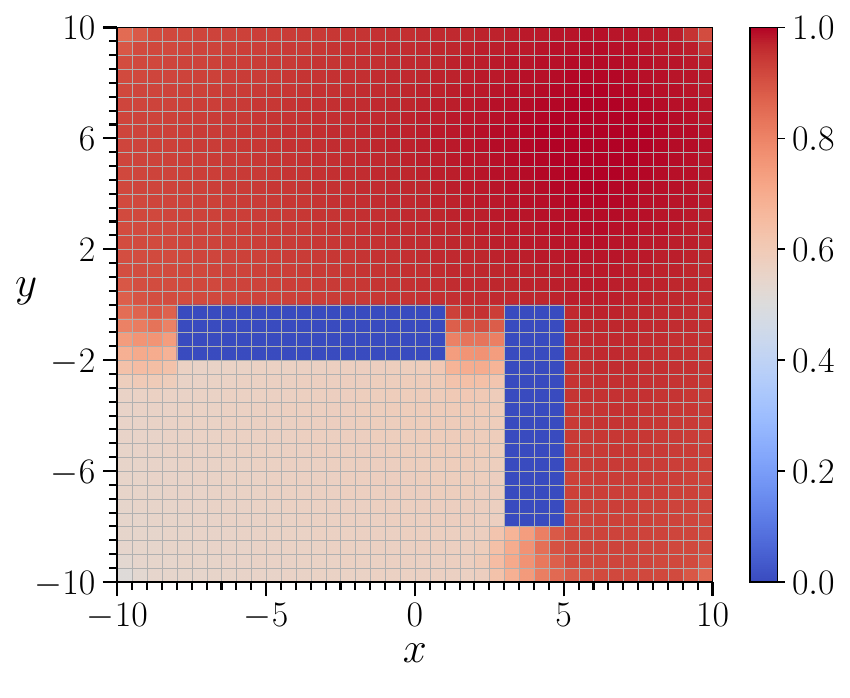}
}\hfill
\subfigure[Scenario approach intervals]{%
\label{fig:results_car:scenario}%
\includegraphics[height=4.2cm]{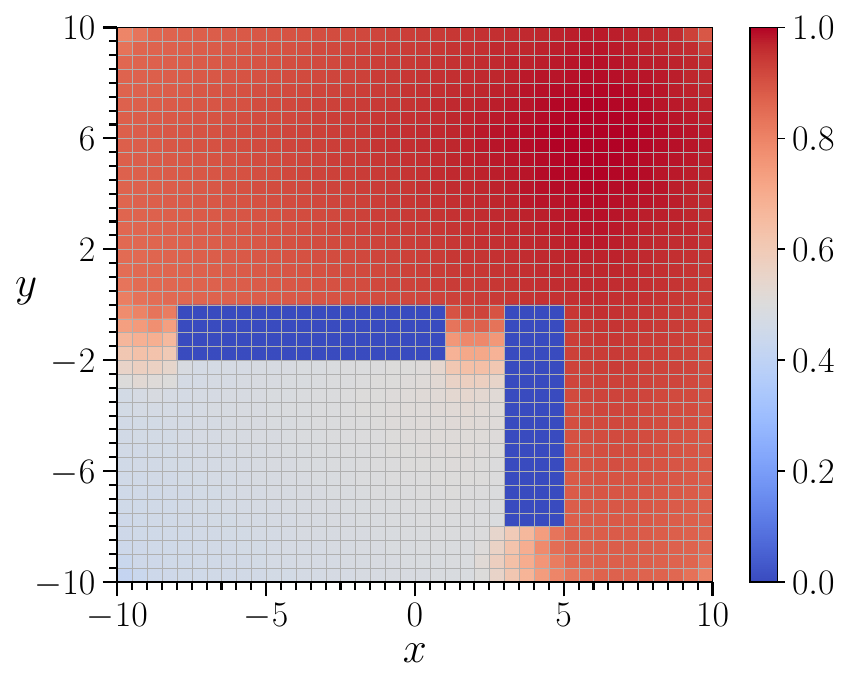}
}\hfill
\subfigure[Simulation under (a)]{%
\label{fig:results_car:traj}%
\includegraphics[height=4.2cm]{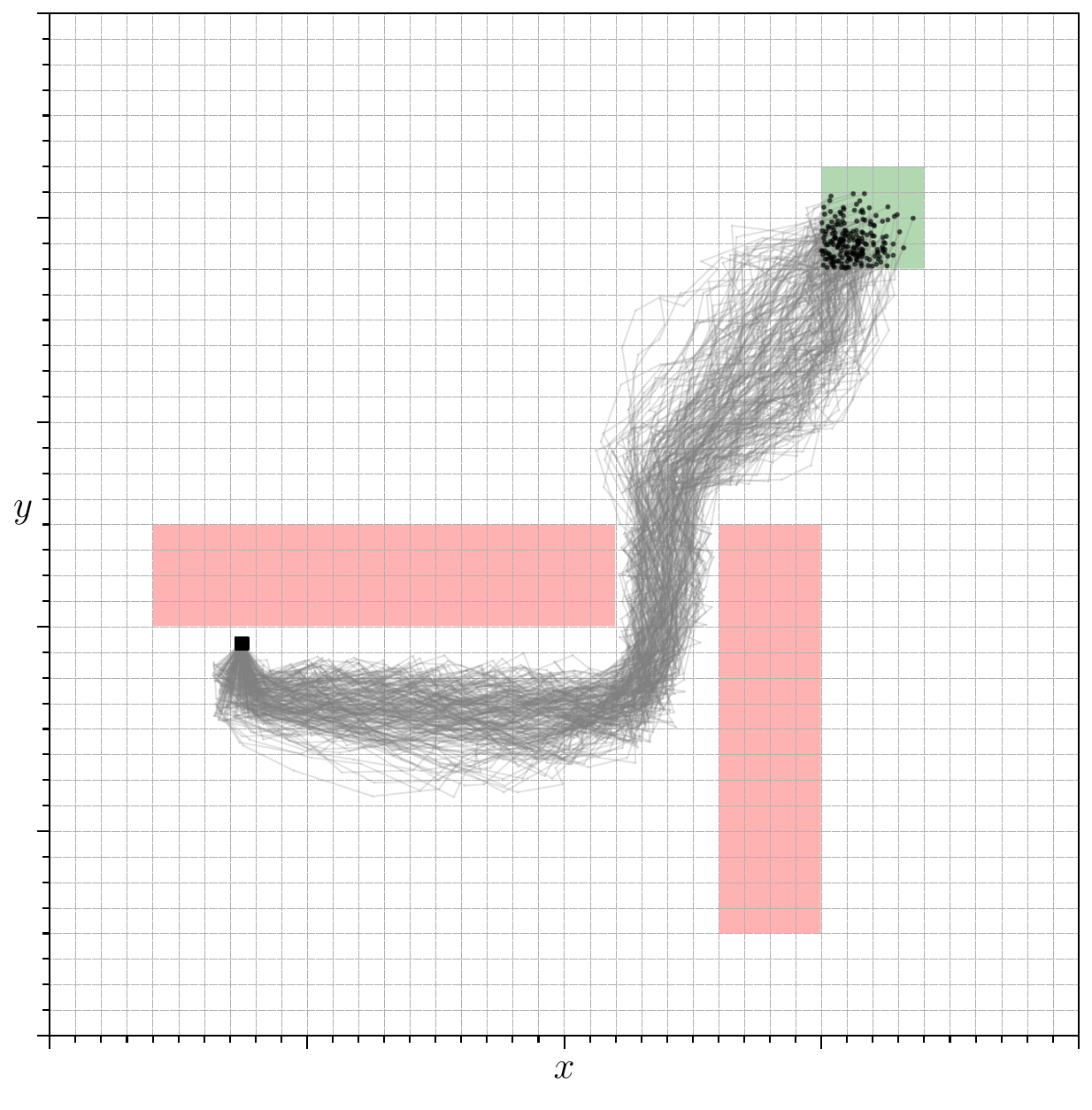}
\vspace{0.5cm}}
}
\end{figure}

\smallskip
\noindent\textbf{Lower bounds on reach-avoid probabilities.\,\,}
We investigate whether our IMDP abstractions lead to sound and non-trivial lower bounds on the reach-avoid probability $\satprob^\system_{\policy}(\xGoal, \xUnsafe,\horizon)$.
A heatmap of these probabilities for the car parking benchmark (probability intervals obtained using \theoremref{thm:PAC_interval} with $N = 10\,000$ samples) is shown in~\figureref{fig:results_car:CP} (results for the other benchmarks and/or a lower number of samples of $N = 1\,000$ are
\ifappendix
    in~\cref{app:Experiments}).
\else
    in~\citet[App.~E]{L4DC_with_appendix}).
\fi
For this case,~\figureref{fig:results_car:traj} shows a simulated trajectory under the resulting policy obtained from~\cref{eq:optimal_policy}.
These results confirm that our method yields reliable policies with non-trivial reach-avoid guarantees in practice.

\smallskip
\noindent\textbf{Comparison to scenario approach.\,\,}
We benchmark our IMDPs with probability intervals from~\theoremref{thm:PAC_interval} against the approach from~\cite{DBLP:conf/aaai/BadingsRA023}, which instead uses the scenario approach.
The resulting reach-avoid probabilities for the car benchmark are in~\figureref{fig:results_car:scenario} (again, 
\ifappendix
    see~\cref{app:Experiments}
\else
    see~\citet[App.~E]{L4DC_with_appendix}
\fi
for the other benchmarks).
Using Clopper-Pearson leads to tighter intervals than those from~\cite{DBLP:conf/aaai/BadingsRA023}, thus leading to policies with better reach-avoid guarantees.

\smallskip
\noindent\textbf{The role of the scaling factors $\lambda_{i \to j}$.\,\,}
Finally, we demonstrate the importance of choosing the scaling factors $\lambda_{i \to j}$ defining the IMDP actions.
To this end, we run all three benchmarks with a smaller $\Lambda$ that upper bounds $\lambda_{i \to j}$, as defined in \cref{sec:underapproximation} (we report the precise values of $\Lambda$
\ifappendix
    in~\cref{app:Experiments}).
\else
    in~\citet[App.~E]{L4DC_with_appendix}).
\fi
Our results, which we present 
\ifappendix
    in~\cref{app:Experiments},
\else
    in~\citet[App.~E]{L4DC_with_appendix},
\fi
demonstrate that a lower $\Lambda$ generally leads to decreased reach-avoid probabilities.
This is likely because a lower $\Lambda$ leads to smaller backward reachable sets, which causes the system taking more steps until it reaches the goal -- investigating this effect in more detail is an important aspect for future research.

\section{Conclusion}
\label{sec:conclusion}
We presented a data-driven approach for the automated synthesis of policies for nonlinear systems with additive stochastic noise. 
Our method only requires samples of the system and the Lipschitz constant, thus overcoming the limitations of model-based abstractions when the dynamics are not fully known.
Our numerical experiments show our approach yields robust and reliable policies. 

This work opens pathways to enhance data-driven methods for controller synthesis in real-world settings where traditional modelling is infeasible.
The main limitations of this work are the restriction to additive and i.i.d. stochastic noise, and the sensitivity of the framework to hyperparameters.
Given the high amount of data required to construct the IMDP, this work only scales to systems with a small number of state variables.
Future directions include investigating tighter bounds for PAC guarantees and integrating our techniques with other frameworks to reduce the computational complexity.
In particular, exploiting structural properties of the dynamics and integrating our abstraction technique with a learning framework are two ideas to overcome the scalability limitations.

\acks{This paper was supported by the Horizon Europe EIC project SymAware (101070802), the ERC grant 101089047, and the EPSRC grant EP/Y028872/1, Mathematical Foundations of Intelligence: An “Erlangen Programme” for AI.}

\bibliography{references.bib}

\begin{thebibliography}{65}
\providecommand{\natexlab}[1]{#1}
\providecommand{\url}[1]{\texttt{#1}}
\expandafter\ifx\csname urlstyle\endcsname\relax
  \providecommand{\doi}[1]{doi: #1}\else
  \providecommand{\doi}{doi: \begingroup \urlstyle{rm}\Url}\fi

\bibitem[Abate et~al.(2008)Abate, Prandini, Lygeros, and Sastry]{DBLP:journals/automatica/AbatePLS08}
Alessandro Abate, Maria Prandini, John Lygeros, and Shankar Sastry.
\newblock Probabilistic reachability and safety for controlled discrete time stochastic hybrid systems.
\newblock \emph{Autom.}, 44\penalty0 (11):\penalty0 2724--2734, 2008.

\bibitem[Abate et~al.(2024)Abate, Giacobbe, and Roy]{DBLP:conf/cav/AbateGR24}
Alessandro Abate, Mirco Giacobbe, and Diptarko Roy.
\newblock Stochastic omega-regular verification and control with supermartingales.
\newblock In \emph{{CAV} {(3)}}, volume 14683 of \emph{LNCS}, pages 395--419. Springer, 2024.

\bibitem[Althoff et~al.(2021)Althoff, Frehse, and Girard]{DBLP:journals/arcras/AlthoffFG21}
Matthias Althoff, Goran Frehse, and Antoine Girard.
\newblock Set propagation techniques for reachability analysis.
\newblock \emph{Annu. Rev. Control. Robotics Auton. Syst.}, 4:\penalty0 369--395, 2021.

\bibitem[Alur et~al.(1998)Alur, Henzinger, Kupferman, and Vardi]{DBLP:conf/concur/AlurHKV98}
Rajeev Alur, Thomas~A. Henzinger, Orna Kupferman, and Moshe~Y. Vardi.
\newblock Alternating refinement relations.
\newblock In \emph{{CONCUR}}, volume 1466 of \emph{LNCS}, pages 163--178. Springer, 1998.

\bibitem[Apostol(1982)]{apostol1981}
Tom.~M. Apostol.
\newblock \emph{Mathematical Analysis}.
\newblock Addison-Wesley Publishing Company, 5th edition, 1982.

\bibitem[Badings et~al.(2022)Badings, Abate, Jansen, Parker, Poonawala, and Stoelinga]{DBLP:conf/aaai/BadingsA00PS22}
Thom~S. Badings, Alessandro Abate, Nils Jansen, David Parker, Hasan~A. Poonawala, and Mari{\"{e}}lle Stoelinga.
\newblock Sampling-based robust control of autonomous systems with non-gaussian noise.
\newblock In \emph{{AAAI}}, pages 9669--9678. {AAAI} Press, 2022.

\bibitem[Badings et~al.(2023{\natexlab{a}})Badings, Romao, Abate, and Jansen]{DBLP:conf/aaai/BadingsRA023}
Thom~S. Badings, Licio Romao, Alessandro Abate, and Nils Jansen.
\newblock Probabilities are not enough: Formal controller synthesis for stochastic dynamical models with epistemic uncertainty.
\newblock In \emph{{AAAI}}, pages 14701--14710. {AAAI} Press, 2023{\natexlab{a}}.

\bibitem[Badings et~al.(2023{\natexlab{b}})Badings, Romao, Abate, Parker, Poonawala, Stoelinga, and Jansen]{DBLP:journals/jair/BadingsRAPPSJ23}
Thom~S. Badings, Licio Romao, Alessandro Abate, David Parker, Hasan~A. Poonawala, Mari{\"{e}}lle Stoelinga, and Nils Jansen.
\newblock Robust control for dynamical systems with non-gaussian noise via formal abstractions.
\newblock \emph{J. Artif. Intell. Res.}, 76:\penalty0 341--391, 2023{\natexlab{b}}.

\bibitem[Badings et~al.(2024)Badings, Romao, Abate, and Jansen]{DBLP:conf/eucc/BadingsRA024}
Thom~S. Badings, Licio Romao, Alessandro Abate, and Nils Jansen.
\newblock A stability-based abstraction framework for reach-avoid control of stochastic dynamical systems with unknown noise distributions.
\newblock In \emph{{ECC}}, pages 564--570. {IEEE}, 2024.

\bibitem[Baier and Katoen(2008)]{BaierKatoen08}
Christel Baier and Joost{-}Pieter Katoen.
\newblock \emph{Principles of model checking}.
\newblock {MIT} Press, 2008.

\bibitem[Bansal et~al.(2017)Bansal, Chen, Herbert, and Tomlin]{DBLP:conf/cdc/BansalCHT17}
Somil Bansal, Mo~Chen, Sylvia~L. Herbert, and Claire~J. Tomlin.
\newblock Hamilton-jacobi reachability: {A} brief overview and recent advances.
\newblock In \emph{{CDC}}, pages 2242--2253. {IEEE}, 2017.

\bibitem[Banse et~al.(2023)Banse, Romao, Abate, and Jungers]{DBLP:journals/corr/abs-2303-17618}
Adrien Banse, Licio Romao, Alessandro Abate, and Rapha{\"{e}}l~M. Jungers.
\newblock Data-driven abstractions via adaptive refinements and a kantorovich metric.
\newblock \emph{CoRR}, abs/2303.17618, 2023.

\bibitem[Belta et~al.(2017)Belta, Yordanov, and Gol]{belta2017formal}
Calin Belta, Boyan Yordanov, and Ebru~Aydin Gol.
\newblock \emph{Formal methods for discrete-time dynamical systems}, volume~15.
\newblock Springer, 2017.

\bibitem[Bertsekas and Shreve(1978)]{Bertsekas.Shreve78}
Dimitri~P. Bertsekas and Steven~E. Shreve.
\newblock \emph{Stochastic Optimal Control: The Discrete-time Case}.
\newblock Athena Scientific, 1978.
\newblock ISBN 1-886529-03-5.

\bibitem[Calbert et~al.(2024)Calbert, Girard, and Jungers]{DBLP:journals/corr/abs-2410-06083}
Julien Calbert, Antoine Girard, and Rapha{\"{e}}l~M. Jungers.
\newblock Classification of simulation relations for symbolic control.
\newblock \emph{CoRR}, abs/2410.06083, 2024.

\bibitem[Campi et~al.(2021)Campi, Car{\`{e}}, and Garatti]{DBLP:journals/arc/CampiCG21}
Marco~C. Campi, Algo Car{\`{e}}, and Simone Garatti.
\newblock The scenario approach: {A} tool at the service of data-driven decision making.
\newblock \emph{Annu. Rev. Control.}, 52:\penalty0 1--17, 2021.

\bibitem[Clopper and Pearson(1934)]{clopper1934use}
Charles~J Clopper and Egon~S Pearson.
\newblock The use of confidence or fiducial limits illustrated in the case of the binomial.
\newblock \emph{Biometrika}, 26\penalty0 (4):\penalty0 404--413, 1934.

\bibitem[Coppola et~al.(2023)Coppola, Peruffo, and Jr.]{DBLP:journals/csysl/CoppolaPM23}
Rudi Coppola, Andrea Peruffo, and Manuel~Mazo Jr.
\newblock Data-driven abstractions for verification of linear systems.
\newblock \emph{{IEEE} Control. Syst. Lett.}, 7:\penalty0 2737--2742, 2023.

\bibitem[Delimpaltadakis et~al.(2023)Delimpaltadakis, Lahijanian, Jr., and Laurenti]{DBLP:conf/hybrid/Delimpaltadakis23}
Giannis Delimpaltadakis, Morteza Lahijanian, Manuel~Mazo Jr., and Luca Laurenti.
\newblock Interval markov decision processes with continuous action-spaces.
\newblock In \emph{{HSCC}}, pages 12:1--12:10. {ACM}, 2023.

\bibitem[Devonport et~al.(2021)Devonport, Saoud, and Arcak]{DBLP:conf/cdc/DevonportSA21}
Alex Devonport, Adnane Saoud, and Murat Arcak.
\newblock Symbolic abstractions from data: {A} {PAC} learning approach.
\newblock In \emph{{CDC}}, pages 599--604. {IEEE}, 2021.

\bibitem[Fan et~al.(2018)Fan, Mathur, Mitra, and Viswanathan]{DBLP:conf/cav/FanMM018}
Chuchu Fan, Umang Mathur, Sayan Mitra, and Mahesh Viswanathan.
\newblock Controller synthesis made real: Reach-avoid specifications and linear dynamics.
\newblock In \emph{{CAV} {(1)}}, volume 10981 of \emph{LNCS}, pages 347--366. Springer, 2018.

\bibitem[Givan et~al.(2000)Givan, Leach, and Dean]{DBLP:journals/ai/GivanLD00}
Robert Givan, Sonia~M. Leach, and Thomas~L. Dean.
\newblock Bounded-parameter markov decision processes.
\newblock \emph{Artif. Intell.}, 122\penalty0 (1-2):\penalty0 71--109, 2000.

\bibitem[Gracia et~al.(2024{\natexlab{a}})Gracia, Boskos, Laurenti, and Lahijanian]{DBLP:conf/l4dc/GraciaBLL24}
Ibon Gracia, Dimitris Boskos, Luca Laurenti, and Morteza Lahijanian.
\newblock Data-driven strategy synthesis for stochastic systems with unknown nonlinear disturbances.
\newblock In \emph{{L4DC}}, volume 242 of \emph{Proceedings of Machine Learning Research}, pages 1633--1645. {PMLR}, 2024{\natexlab{a}}.

\bibitem[Gracia et~al.(2024{\natexlab{b}})Gracia, Laurenti, Jr., Abate, and Lahijanian]{DBLP:journals/corr/abs-2412-11343}
Ibon Gracia, Luca Laurenti, Manuel~Mazo Jr., Alessandro Abate, and Morteza Lahijanian.
\newblock Temporal logic control for nonlinear stochastic systems under unknown disturbances.
\newblock \emph{CoRR}, abs/2412.11343, 2024{\natexlab{b}}.

\bibitem[Hashimoto et~al.(2022)Hashimoto, Saoud, Kishida, Ushio, and Dimarogonas]{DBLP:journals/automatica/HashimotoSKUD22}
Kazumune Hashimoto, Adnane Saoud, Masako Kishida, Toshimitsu Ushio, and Dimos~V. Dimarogonas.
\newblock Learning-based symbolic abstractions for nonlinear control systems.
\newblock \emph{Autom.}, 146:\penalty0 110646, 2022.

\bibitem[Hermanns et~al.(2011)Hermanns, Parma, Segala, Wachter, and Zhang]{DBLP:journals/iandc/HermannsPSWZ11}
Holger Hermanns, Augusto Parma, Roberto Segala, Bj{\"{o}}rn Wachter, and Lijun Zhang.
\newblock Probabilistic logical characterization.
\newblock \emph{Inf. Comput.}, 209\penalty0 (2):\penalty0 154--172, 2011.

\bibitem[Iyengar(2005)]{DBLP:journals/mor/Iyengar05}
Garud~N. Iyengar.
\newblock Robust dynamic programming.
\newblock \emph{Math. Oper. Res.}, 30\penalty0 (2):\penalty0 257--280, 2005.

\bibitem[Jackson et~al.(2021)Jackson, Laurenti, Frew, and Lahijanian]{DBLP:conf/hybrid/JacksonLFL21}
John Jackson, Luca Laurenti, Eric~W. Frew, and Morteza Lahijanian.
\newblock Strategy synthesis for partially-known switched stochastic systems.
\newblock In \emph{{HSCC}}, pages 6:1--6:11. {ACM}, 2021.

\bibitem[Kazemi et~al.(2022)Kazemi, Majumdar, Salamati, Soudjani, and Wooding]{DBLP:journals/corr/abs-2206-08069}
Milad Kazemi, Rupak Majumdar, Mahmoud Salamati, Sadegh Soudjani, and Ben Wooding.
\newblock Data-driven abstraction-based control synthesis.
\newblock \emph{CoRR}, abs/2206.08069, 2022.

\bibitem[Khalil and Grizzle(2002)]{khalil2002nonlinear}
Hassan~K Khalil and Jessy~W Grizzle.
\newblock \emph{Nonlinear systems}, volume~3.
\newblock Prentice hall Upper Saddle River, NJ, 2002.

\bibitem[Kwiatkowska et~al.(2011)Kwiatkowska, Norman, and Parker]{DBLP:conf/cav/KwiatkowskaNP11}
Marta~Z. Kwiatkowska, Gethin Norman, and David Parker.
\newblock {PRISM} 4.0: Verification of probabilistic real-time systems.
\newblock In \emph{{CAV}}, volume 6806 of \emph{LNCS}, pages 585--591. Springer, 2011.

\bibitem[Lahijanian et~al.(2015)Lahijanian, Andersson, and Belta]{DBLP:journals/tac/LahijanianAB15}
Morteza Lahijanian, Sean~B. Andersson, and Calin Belta.
\newblock Formal verification and synthesis for discrete-time stochastic systems.
\newblock \emph{{IEEE} Trans. Autom. Control.}, 60\penalty0 (8):\penalty0 2031--2045, 2015.

\bibitem[Larsen and Skou(1991)]{DBLP:journals/iandc/LarsenS91}
Kim~Guldstrand Larsen and Arne Skou.
\newblock Bisimulation through probabilistic testing.
\newblock \emph{Inf. Comput.}, 94\penalty0 (1):\penalty0 1--28, 1991.

\bibitem[Lavaei et~al.(2022{\natexlab{a}})Lavaei, Soudjani, Abate, and Zamani]{LSAZ21}
Abolfazl Lavaei, Sadegh Soudjani, Alessandro Abate, and Majid Zamani.
\newblock Automated verification and synthesis of stochastic hybrid systems: {A} survey.
\newblock \emph{Autom.}, 146:\penalty0 110617, 2022{\natexlab{a}}.

\bibitem[Lavaei et~al.(2022{\natexlab{b}})Lavaei, Soudjani, Frazzoli, and Zamani]{lavaei2022constructing}
Abolfazl Lavaei, Sadegh Soudjani, Emilio Frazzoli, and Majid Zamani.
\newblock Constructing mdp abstractions using data with formal guarantees.
\newblock \emph{IEEE Control Systems Letters}, 7:\penalty0 460--465, 2022{\natexlab{b}}.

\bibitem[Lavaei et~al.(2023)Lavaei, Soudjani, and Frazzoli]{lavaei2023compositional}
Abolfazl Lavaei, Sadegh Soudjani, and Emilio Frazzoli.
\newblock A compositional dissipativity approach for data-driven safety verification of large-scale dynamical systems.
\newblock \emph{{IEEE} Trans. Autom. Control.}, 68\penalty0 (12):\penalty0 7240--7253, 2023.

\bibitem[Makdesi et~al.(2021)Makdesi, Girard, and Fribourg]{DBLP:conf/adhs/MakdesiGF21}
Anas Makdesi, Antoine Girard, and Laurent Fribourg.
\newblock Efficient data-driven abstraction of monotone systems with disturbances.
\newblock In \emph{{ADHS}}, volume~54 of \emph{IFAC-PapersOnLine}, pages 49--54. Elsevier, 2021.

\bibitem[Mathiesen et~al.(2023)Mathiesen, Calvert, and Laurenti]{DBLP:journals/csysl/MathiesenCL23}
Frederik~Baymler Mathiesen, Simeon~C. Calvert, and Luca Laurenti.
\newblock Safety certification for stochastic systems via neural barrier functions.
\newblock \emph{{IEEE} Control. Syst. Lett.}, 7:\penalty0 973--978, 2023.

\bibitem[Mathiesen et~al.(2024)Mathiesen, Haesaert, and Laurenti]{mathiesen2024}
Frederik~Baymler Mathiesen, Sofie Haesaert, and Luca Laurenti.
\newblock Scalable control synthesis for stochastic systems via structural imdp abstractions, 2024.

\bibitem[Meggendorfer et~al.(2024)Meggendorfer, Weininger, and Wienh{\"{o}}ft]{DBLP:journals/corr/abs-2404-05424}
Tobias Meggendorfer, Maximilian Weininger, and Patrick Wienh{\"{o}}ft.
\newblock What are the odds? improving the foundations of statistical model checking.
\newblock \emph{CoRR}, abs/2404.05424, 2024.

\bibitem[Mitchell(2007)]{DBLP:conf/hybrid/Mitchell07}
Ian~M. Mitchell.
\newblock Comparing forward and backward reachability as tools for safety analysis.
\newblock In \emph{{HSCC}}, volume 4416 of \emph{LNCS}, pages 428--443. Springer, 2007.

\bibitem[Nejati et~al.(2023)Nejati, Lavaei, Jagtap, Soudjani, and Zamani]{DBLP:journals/tac/NejatiLJSZ23}
Ameneh Nejati, Abolfazl Lavaei, Pushpak Jagtap, Sadegh Soudjani, and Majid Zamani.
\newblock Formal verification of unknown discrete- and continuous-time systems: {A} data-driven approach.
\newblock \emph{{IEEE} Trans. Autom. Control.}, 68\penalty0 (5):\penalty0 3011--3024, 2023.

\bibitem[Newcombe(1998)]{newcombe1998two}
Robert~G Newcombe.
\newblock Two-sided confidence intervals for the single proportion: comparison of seven methods.
\newblock \emph{Statistics in medicine}, 17\penalty0 (8):\penalty0 857--872, 1998.

\bibitem[Nilim and Ghaoui(2005)]{DBLP:journals/ior/NilimG05}
Arnab Nilim and Laurent~El Ghaoui.
\newblock Robust control of markov decision processes with uncertain transition matrices.
\newblock \emph{Oper. Res.}, 53\penalty0 (5):\penalty0 780--798, 2005.

\bibitem[Peruffo and Mazo(2023)]{DBLP:journals/csysl/PeruffoM23}
Andrea Peruffo and Manuel Mazo.
\newblock Data-driven abstractions with probabilistic guarantees for linear {PETC} systems.
\newblock \emph{{IEEE} Control. Syst. Lett.}, 7:\penalty0 115--120, 2023.

\bibitem[Pnueli(1977)]{DBLP:conf/focs/Pnueli77}
Amir Pnueli.
\newblock The temporal logic of programs.
\newblock In \emph{{FOCS}}, pages 46--57. {IEEE} Computer Society, 1977.

\bibitem[Puterman(1994)]{DBLP:books/wi/Puterman94}
Martin~L. Puterman.
\newblock \emph{Markov Decision Processes: Discrete Stochastic Dynamic Programming}.
\newblock Wiley Series in Probability and Statistics. Wiley, 1994.

\bibitem[Reissig et~al.(2017)Reissig, Weber, and Rungger]{DBLP:journals/tac/ReissigWR17}
Gunther Reissig, Alexander Weber, and Matthias Rungger.
\newblock Feedback refinement relations for the synthesis of symbolic controllers.
\newblock \emph{{IEEE} Trans. Autom. Control.}, 62\penalty0 (4):\penalty0 1781--1796, 2017.

\bibitem[Rober et~al.(2022)Rober, Katz, Sidrane, Yel, Everett, Kochenderfer, and How]{DBLP:journals/corr/abs-2209-14076}
Nicholas Rober, Sydney~M. Katz, Chelsea Sidrane, Esen Yel, Michael Everett, Mykel~J. Kochenderfer, and Jonathan~P. How.
\newblock Backward reachability analysis of neural feedback loops: Techniques for linear and nonlinear systems.
\newblock \emph{CoRR}, abs/2209.14076, 2022.

\bibitem[Romao et~al.(2023)Romao, Papachristodoulou, and Margellos]{romao2022tac}
Licio Romao, Antonis Papachristodoulou, and Kostas Margellos.
\newblock On the exact feasibility of convex scenario programs with discarded constraints.
\newblock \emph{{IEEE} Trans. Autom. Control.}, 68\penalty0 (4):\penalty0 1986--2001, 2023.

\bibitem[Salamati et~al.(2024)Salamati, Lavaei, Soudjani, and Zamani]{DBLP:journals/automatica/SalamatiLSZ24}
Ali Salamati, Abolfazl Lavaei, Sadegh Soudjani, and Majid Zamani.
\newblock Data-driven verification and synthesis of stochastic systems via barrier certificates.
\newblock \emph{Autom.}, 159:\penalty0 111323, 2024.

\bibitem[Salamon(2016)]{Salamon16}
Dietmar Salamon.
\newblock \emph{Measure and Integration}.
\newblock European Mathematical Society, 2016, 2016.

\bibitem[Sch{\"o}n et~al.(2024)Sch{\"o}n, Naseer, Wooding, and Soudjani]{schon2024data}
Oliver Sch{\"o}n, Shammakh Naseer, Ben Wooding, and Sadegh Soudjani.
\newblock Data-driven abstractions via binary-tree {G}aussian processes for formal verification.
\newblock \emph{IFAC-PapersOnLine}, 58\penalty0 (11):\penalty0 115--122, 2024.

\bibitem[Soudjani and Abate(2013)]{DBLP:journals/siamads/SoudjaniA13}
Sadegh Esmaeil~Zadeh Soudjani and Alessandro Abate.
\newblock Adaptive and sequential gridding procedures for the abstraction and verification of stochastic processes.
\newblock \emph{{SIAM} J. Appl. Dyn. Syst.}, 12\penalty0 (2):\penalty0 921--956, 2013.

\bibitem[Soudjani et~al.(2015)Soudjani, Gevaerts, and Abate]{FAUST15}
Sadegh Esmaeil~Zadeh Soudjani, Caspar Gevaerts, and Alessandro Abate.
\newblock {FAUST} \({}^{\mbox{ 2}}\) : Formal abstractions of uncountable-state stochastic processes.
\newblock In \emph{{TACAS}}, volume 9035 of \emph{LNCS}, pages 272--286. Springer, 2015.

\bibitem[Stipanovic et~al.(2003)Stipanovic, Hwang, and Tomlin]{DBLP:conf/eucc/StipanovicHT03}
Dusan~M. Stipanovic, Inseok Hwang, and Claire~J. Tomlin.
\newblock Computation of an over-approximation of the backward reachable set using subsystem level set functions.
\newblock In \emph{{ECC}}, pages 300--305. {IEEE}, 2003.

\bibitem[Suilen et~al.(2025)Suilen, Badings, Bovy, Parker, and Jansen]{Suilen2025}
Marnix Suilen, Thom Badings, Eline~M. Bovy, David Parker, and Nils Jansen.
\newblock \emph{Robust Markov Decision Processes: A Place Where AI and Formal Methods Meet}, pages 126--154.
\newblock Springer Nature Switzerland, Cham, 2025.
\newblock ISBN 978-3-031-75778-5.
\newblock \doi{10.1007/978-3-031-75778-5\\_7}.

\bibitem[Summers and Lygeros(2010)]{DBLP:journals/automatica/SummersL10}
Sean Summers and John Lygeros.
\newblock Verification of discrete time stochastic hybrid systems: {A} stochastic reach-avoid decision problem.
\newblock \emph{Autom.}, 46\penalty0 (12):\penalty0 1951--1961, 2010.

\bibitem[Tabuada(2009)]{DBLP:books/daglib/0032856}
Paulo Tabuada.
\newblock \emph{Verification and Control of Hybrid Systems - {A} Symbolic Approach}.
\newblock Springer, 2009.

\bibitem[Thulin(2014)]{Thulin_2014}
Måns Thulin.
\newblock The cost of using exact confidence intervals for a binomial proportion.
\newblock \emph{Electronic Journal of Statistics}, 8\penalty0 (1), January 2014.
\newblock ISSN 1935-7524.
\newblock \doi{10.1214/14-ejs909}.

\bibitem[van Huijgevoort et~al.(2023)van Huijgevoort, Sch{\"{o}}n, Soudjani, and Haesaert]{Huijgevoort2023SySCoRe}
Birgit van Huijgevoort, Oliver Sch{\"{o}}n, Sadegh Soudjani, and Sofie Haesaert.
\newblock Syscore: Synthesis via stochastic coupling relations.
\newblock In \emph{{HSCC}}, pages 13:1--13:11. {ACM}, 2023.

\bibitem[Wolff et~al.(2012)Wolff, Topcu, and Murray]{DBLP:conf/cdc/WolffTM12}
Eric~M. Wolff, Ufuk Topcu, and Richard~M. Murray.
\newblock Robust control of uncertain markov decision processes with temporal logic specifications.
\newblock In \emph{{CDC}}, pages 3372--3379. {IEEE}, 2012.

\bibitem[Yang et~al.(2022)Yang, Zhang, Jeannin, and Ozay]{DBLP:journals/tcad/YangZJO22}
Liren Yang, Hang Zhang, Jean{-}Baptiste Jeannin, and Necmiye Ozay.
\newblock Efficient backward reachability using the minkowski difference of constrained zonotopes.
\newblock \emph{{IEEE} Trans. Comput. Aided Des. Integr. Circuits Syst.}, 41\penalty0 (11):\penalty0 3969--3980, 2022.

\bibitem[Yin et~al.(2019)Yin, Packard, Arcak, and Seiler]{DBLP:conf/amcc/YinPAS19}
He~Yin, Andrew~K. Packard, Murat Arcak, and Peter~J. Seiler.
\newblock Finite horizon backward reachability analysis and control synthesis for uncertain nonlinear systems.
\newblock In \emph{{ACC}}, pages 5020--5026. {IEEE}, 2019.

\bibitem[Zikelic et~al.(2023)Zikelic, Lechner, Henzinger, and Chatterjee]{DBLP:conf/aaai/ZikelicLHC23}
Dorde Zikelic, Mathias Lechner, Thomas~A. Henzinger, and Krishnendu Chatterjee.
\newblock Learning control policies for stochastic systems with reach-avoid guarantees.
\newblock In \emph{{AAAI}}, pages 11926--11935. {AAAI} Press, 2023.

\end{thebibliography}

\ifappendix%
\appendix%
\section {Correctness of the abstraction}
\label{app:correctness}

In this appendix, we show the correctness of the IMDP abstraction defined in \cref{sec:abstraction} for solving \cref{prob:Problem1}.
This correctness proof is based on an extension of an \emph{alternating simulation relation} for stochastic systems.
Alternating simulation relations were originally developed by~\cite{DBLP:conf/concur/AlurHKV98} as a variant of simulation relations that can be used to solve control problems~\citep{DBLP:books/daglib/0032856}.
Alternating simulation relations are also closely related to \emph{feedback refinement relations}~\citep{DBLP:journals/tac/ReissigWR17}, and we refer to the work by~\cite{DBLP:journals/corr/abs-2410-06083} for a more detailed classification of such relations.
The behavioural relation defined below can also be seen as a variant of the notions defined by~\cite{DBLP:journals/iandc/HermannsPSWZ11} and \cite{DBLP:journals/iandc/LarsenS91} for stochastic systems.

Intuitively, we want to certify that all behaviours from one model (the abstract IMDP) can be \emph{matched} by another model (the dynamical system) \emph{under some policy}.
This property is captured by the following~definition.

\begin{definition}[Probabilistic alternating simulation relation~\citep{DBLP:conf/eucc/BadingsRA024}]
    \label{def:relation}
    A function $\mathcal{T}: \RR^n \to \States$ induces a \emph{probabilistic alternating simulation relation} from an IMDP $\imdp = \IMDP$ to a system $\system$ as in~\cref{eq:DTSS}~if
    \begin{enumerate}
        \setlength\itemsep{0em}
        \item[(1)] for the initial states, we have $\initState = \mathcal{T}(x_I)$, and
        \item[(2)] for all $x \in \RR^n$ with $s_i = \mathcal{T}(x)$ and for all $a_j \in \Actions(s_i)$, there exists an input $u \in \mathcal{U}$ such that
        \begin{equation}
            \label{}
            \qquad\enskip \eta(\hat{x}, \mathcal{T}^{-1}(s')) \in \transfuncImdp(s_i,a_j)(s'), \quad \forall s' \in \States, \, \forall \hat{x} \in R_j(\lambda_{i \to j}).
        \end{equation}
    \end{enumerate}
\end{definition}

\noindent
Condition~(2) requires that, for all $x \in \RR^n$ and $a_j \in \Actions(s_i)$, where $s_i = \mathcal{T}(x)$, there exists an input $u \in \mathcal{U}$ such that all possible probabilistic behaviours of the system $\system$ is (informally speaking) contained in that of $\imdp$.
We use the common notation of writing $\imdp \preceq_\mathcal{T} \system$ if \definitionref{def:relation} holds~\citep{DBLP:conf/concur/AlurHKV98}.

\begin{remark}[Comparison of relations]
Let us make the following remarks on \definitionref{def:relation}:
\begin{enumerate}
    \item Simulation relations are usually defined with a binary relation $R \subset \RR^n \times \States$ between the two models. Here, we directly define the relation using the abstraction function $\mathcal{T}$. Note, however, that the abstraction function uniquely generates a binary relation defined as $R = \{ (x,s) \in \RR^n \times \States : \mathcal{T}(x) = s \}$.
    \item A straight extension of the alternating simulation relation from~\citet[Def.~4.19]{DBLP:books/daglib/0032856} would require that for all $a_j \in \Actions(s_i)$, there exists an input $u \in \mathcal{U}$ such that the probability distributions over successor states (viewed through the abstraction function $\mathcal{T}$) are equivalent. As our abstract model is an IMDP, we instead require that the probability distribution over successor states in the concrete system $\system$ \emph{is contained} in the probability intervals of the IMDP.
    \item In the scope of safety problems, one is typically interested in showing that $\system \preceq_\mathcal{T} \imdp$, i.e., all behaviours of the concrete system is contained in the abstract IMDP. Note, however, that here we require that $\imdp \preceq_\mathcal{T} \system$, because for the reach-avoid specifications we consider, all possible behaviours of the IMDP need to be matched by the concrete system.
\end{enumerate}
\end{remark}

\begin{lemma}
    The IMDP abstraction $\imdp = \IMDP$ obtained from the abstraction function $\mathcal{T} \colon \RR^n \to \States$ induces a probabilistic alternating simulation relation from $\imdp$ to $\system$, i.e., $\imdp \preceq_\mathcal{T} \system$.
\end{lemma}

\begin{proof}
    Condition (1) in \definitionref{def:relation} is satisfied by the definition of the initial state of the IMDP.
    For condition (2), pick any $x \in \RR^n$ and any $a_j \in \Actions(s_i)$, where $s_i = \mathcal{T}(x)$.
    By construction, action $a_j$ enabled in state $s_i$ implies the existence of an input $u \in \mathcal{U}$ such that $f(x,u) \in R_j(\lambda_{i \to j})$.
    Furthermore, for every $s' \in \States$ and $\hat{x} \in R_j(\lambda_{i \to j})$, it must hold that $\eta(\hat{x}, \mathcal{T}^{-1}(s')) \in \transfuncImdp(s_i,a_j)(s')$, which is satisfied by taking the min/max over $\eta$ in \cref{eq:transition_function}.
    Thus, the claim follows.
\end{proof}

The existence of a probabilistic alternating simulation relation can be used to solve \cref{prob:Problem1} based on the finite abstraction, as stated in the following theorem.

\begin{theorem}[Policy synthesis~\citep{DBLP:journals/jair/BadingsRAPPSJ23}]
    \label{thm:existence_controller_appendix}
    Let $\imdp$ be the IMDP abstraction for the dynamical system $\system$.
    For every IMDP scheduler $\scheduler \in \schedulerSpace^\imdp$, there exists a policy $\policy \in \policySpace^\system$ for $\system$~such~that
    \begin{equation}
        \label{eq:existence_controller_appendix}
        \min_{P \in \transfuncImdp} \satprob^\imdp_{\scheduler,P}(\sGoal, \sUnsafe,\horizon)
        \, \leq \,
        \satprob^\system_{\policy}(\xGoal, \xUnsafe,\horizon).
    \end{equation}
\end{theorem} 

The proof of \theoremref{thm:existence_controller_appendix} uses the fact that the existence of a probabilistic alternating simulation relation $\imdp \preceq_\mathcal{T} \system$ preserves the satisfaction of \emph{probabilistic computation tree logic} (PCTL) specifications, which subsume reach-avoid specifications~\citep{DBLP:journals/iandc/HermannsPSWZ11}.
For MDP abstractions, this preservation of satisfaction probabilities holds with equality, whereas for IMDPs (like we have), this holds with inequality as in \cref{eq:existence_controller_appendix}.
For further details, we refer to~\cite{DBLP:journals/jair/BadingsRAPPSJ23}.

\section{Proof of \theoremref{thm:Jacobian}}
\label{app:Jacobian}

Let us denote by $f^{(i)}$ the $i^\text{th}$ component of the vector function $f \colon \RR^n \times \mathcal{U} \to \RR^n$.
Using the mean value theorem~\citep{apostol1981}, we know there exists $c \in [0, 1]$ such that
\begin{equation*}
    f(x_1, u_\ell)^{(i)} - f(x_2, u_\ell)^{(i)} = (x_1 - x_2)^\top \nabla{f^{(i)}(c x_1 + (1-c)x_2, u_\ell)}.   
\end{equation*}
Consequently, it holds that
\begin{equation*}
\begin{split}
    |f(x_1, u_\ell)^{(i)} - f(x_2, u_\ell)^{(i)}| &= |(x_1 - x_2)^\top \nabla{f^{(i)}(c x_1 + (1-c)x_2, u_\ell)}|
    \\
    &= \sum_{j = 1}^{n} |(x_1 - x_2)^{(j)} \cdot\nabla{f^{(i)}(c x_1 + (1-c)x_2, u_\ell)}^{(j)}|   
    \\
    &\leq \sum_{j = 1}^{n} |x_1 - x_2|^{(j)} \cdot |\nabla{f^{(i)}(c x_1 + (1-c)x_2, u_\ell)}|^{(j)}  
    \\
    &= \sum_{j = 1}^{n} |x_1 - x_2|^{(j)} \cdot \max_{c \in [0, 1]}|\nabla{f^{(i)}(c x_1 + (1-c)x_2, u_\ell)}|^{(j)}
    \\
    &\leq \sum_{j = 1}^{n} |x_1 - x_2|^{(j)} \cdot \max_{x \in R_l}|\nabla{f^{(i)}(x, u_\ell)}|^{(j)}
    \\
    &\leq |x_1 - x_2|^\top \max_{x \in R_l}{|\nabla{f^{(i)}(x, u_\ell)}|}.
\end{split}
\end{equation*}
In other words, we have for all $i = 1,\ldots,n$ that
\begin{equation}
    \label{eq:proof_Jacobian}
    |f(x_1, u_\ell)^{(i)} - f(x_2, u_\ell)^{(i)}| \leq |x_1 - x_2|^\top \max_{x \in R_l}{|\nabla{f^{(i)}(x, u_\ell)}|}.
\end{equation}
Generalising \cref{eq:proof_Jacobian} to all $n$ dimensions of the vector function $f$ yields \theoremref{thm:Jacobian}, which concludes the proof.

\section{Proof of \theoremref{thm:underapproximation}}
\label{app:underapproximation}
We will use \cref{thm:Jacobian} to show that, for every $y \in A_j(x)$, it holds that $y \in \Pre(R_j(\lambda_{i\rightarrow j}))$.
Specifically, we have that
\begin{equation*}
\begin{split}
    y \in A_j(x) \Leftrightarrow \big\lVert J^+(R_i) \cdot | x - y | \big\rVert_\infty \leq r_j(x', \lambda_{i \to j})
    \\
    \Rightarrow \big\lVert f(x, u_\ell) - f(y, u_\ell)\big\rVert_\infty \leq \big\lVert J^+(R_i) \cdot | x - y | \big\rVert_\infty \leq r_j(x', \lambda_{i \to j}) 
    \\
    \Rightarrow \big\lVert x' - f(y, u_\ell)\big\rVert_\infty \leq r_j(x', \lambda_{i \to j}) 
    \\
    \Leftrightarrow f(y, u_\ell) \in B_{r_j(x', \lambda_{i \to j})}^\infty(x')
    \\
    \Rightarrow f(y, u_\ell) \in R_j(\lambda_{i \to j})
    \Rightarrow y \in \Pre(R_j(\lambda_{i \to j})).
\end{split}
\end{equation*}
Thus, $A_j(x) \subseteq \Pre(R_j(\lambda_{i \to j}))$, which concludes the proof.

\section{Details of Computing Probability Intervals}
\label{app:probability_intervals}

Recall from \cref{sec:probability_intervals} that $\check{N}_{i,j}(s')$ and $\hat{N}_{i,j}(s')$ are samples of the binomial distributions $\Bin(N,\check{P}_{i,j}(s'))$ and $\Bin(N,\hat{P}_{i,j}(s'))$, respectively.
Equivalently, we can write $\check{N}_{i,j}(s')$ and $\hat{N}_{i,j}(s')$ in a more explicit form, by \emph{counting} the number of \emph{successes} for a given set of noise samples $\{w^{(1)},\ldots,w^{(N)}\} \in \Omega^N$.
This leads to the following more explicit definition of these quantities.

\begin{definition}[Counting samples]
    \label{def:counting_samples}
    Let $\{w^{(1)},\ldots,w^{(N)}\} \in \Omega^N$ be a set of $N \in \NN$ i.i.d. samples from the noise. We define the sample counts $\check{N}_{i,j}(s')$ and $\hat{N}_{i,j}(s')$ as follows:
    \begin{align*}
    \check{N}_{i,j}(s') &= \big| \big\{ \ell \in \{1,\ldots,N\} : R_j(\lambda_{i \to j}) + w^{(\ell)} \subset R_\ell \big\} \big|,
    \\
    \hat{N}_{i,j}(s') &= \big| \big\{ \ell \in \{1,\ldots,N\} : R_j(\lambda_{i \to j}) + w^{(\ell)} \cap R_\ell \neq \emptyset \big\} \big|.
\end{align*}
\end{definition}

In practice, we can thus count the numbers of samples fully contained in $R_\ell$ (giving $\check{N}_{i,j}(s')$) and those with nonempty intersection with $R_\ell$ (giving $\hat{N}_{i,j}(s'))$), which we then plug into \theoremref{thm:PAC_interval}.

\section{Experiment Details}
\label{app:Experiments}

In this appendix, we give further details about the dynamics and results of our numerical experiments.
We remark that even though we define the distribution of the stochastic noise in each benchmark explicitly, our abstraction technique only requires sampling access to this distribution.
Moreover, even though we use uniform and Gaussian noise distributions in these benchmarks, our abstraction technique can handle any distribution satisfying \cref{assumptions}.

\subsection{Car parking}

\paragraph{System's dynamics.} We consider a car with nonlinear control in the 2D plane.     
The state variables are the $x$ and $y$ position of a car, such that the state at discrete time $k$ is $[x_k, y_k]^\top \in \RR^2$. The velocity and angle of the car can be controlled separately, using $u_k = [v_k, \theta_k] \in \mathcal{U} = [-0.1, 0.1] \times [-\pi, \pi]$. The dynamics of the system in discrete time are defined as
\begin{align*}
    x_{k+1} &= x_k + 10 \delta v_k \cos(\theta_k) + \zeta_1,
    \\
    y_{k+1} &= y_k + 10 \delta v_k \sin(\theta_k) + \zeta_2.
\end{align*}
The noise is sampled from a uniform distribution with bounds $\zeta \in U([-0.55, 0.55] \times [-0.55, 0.55])$.
Note that the dynamics are linear in the state but nonlinear in the inputs.

\paragraph{Reach-avoid specification.}
We synthesise a controller for the following reach-avoid specification over an infinite horizon, $\horizon = \infty$. The goal set is $\xGoal = [5, 7] \times [5, 7]$, and there are three unsafe sets:
\begin{align*}
    \xUnsafe^1 &= \RR^2 \setminus ([-10, 10] \times [-10, 10])
    \\
    \xUnsafe^2 &= [-8, 1] \times [-2, 0]
    \\
    \xUnsafe^3 &= [3, 5] \times [-8, 0],
\end{align*}
such that $\xUnsafe = \xUnsafe^1 \cup \xUnsafe^2 \cup \xUnsafe^3$.
The first unsafe set represents leaving the bounded portion of the state space $[-10,10] \times [-10,10]$, while the other two represent obstacles.

\paragraph{Abstraction.} 
We use a uniform partition with 40 $\times$ 40 states. In each state, we use 7 $\times$ 7 state samples, 7 $\times$ 21 control samples per state sample, and 7 $\times$ 7 voxels. 
We use a maximum scaling factor of $\Lambda = 1.5$ to compute the enabled actions and $N = 10\,000$ noise samples to compute the probability intervals using \theoremref{thm:PAC_interval}. 
The resulting IMDP has $1\,602$ states, $17\,770$ actions, and $452\,574$ transitions. 
Generating the abstraction takes approximately $24$ minutes, and computing an optimal policy using PRISM takes approximately $6$ seconds.

\subsection{Inverted pendulum}

\paragraph{System's dynamics.}
We consider the classical inverted pendulum benchmark.
The two-dimensional state $[\theta_k, \omega_k]^\top \in \RR^2$ models the angle $\theta_k$ and angular velocity $\omega_k$ of the pendulum at time step $k$, and the torque is constrained to $u_k \in \mathcal{U} = [-17.5, 17.5]$.
The dynamics are defined as
\begin{align*}
    \theta_{k+1} &= \theta_k + \delta \omega_k + \zeta_1,
    \\
    \omega_{k+1} &= \omega_k + \delta (-\frac{g}{l}) \sin(-\theta_k) + \frac{u}{m \cdot l^2} + \zeta_2,
\end{align*}
where $\delta = 0.1$s is the time discretization step, $g = 9.81\text{N}\cdot\frac{\text{m}^2}{\text{kg}^2}$ is the gravitational constant, and $l=1$m and $m=1$kg are the length and mass of the pendulum, respectively. The noise $\zeta \sim U( [-0.1, 0.1] \times [-0.2, 0.2] )$ is sampled from a uniform distribution.

\paragraph{Reach-avoid specification.}
We consider the infinite-horizon reach-avoid task to reach a state in $\xGoal = [-0.2, 0.2] \times [-0.4, 0.4]$ while avoiding states in $\xUnsafe = \RR^2 \setminus ([-\pi, \pi] \times [-2, 2])$, which models avoiding angular velocities above $+2$ or below $-2$. 

\paragraph{Abstraction.}
We create an abstraction based on a uniform partition into $32 \times 10$ states. In each state, we use 15 $\times$ 21 state samples, 15 $\times$ 21 control samples per state sample, and 15 $\times$ 15 voxels.
We use a maximum scaling factor of $\Lambda = 2.0$ to compute the enabled actions and $N = 10\,000$ noise samples to compute the probability intervals using \theoremref{thm:PAC_interval}.
The resulting IMDP has $322$ states, $4\,599$ actions, and $108\,193$ transitions.
Generating the abstraction takes approximately $10$ minutes, and computing an optimal policy using PRISM takes less than $1$ second.

\subsection{Harmonic oscillator with nonlinear damping}

\paragraph{System's dynamics.}
We study a harmonic oscillator with nonlinear damping.
The two-dimensional state $[x_k, v_k]^\top \in \RR^n$ models the position $x_k$ and velocity $v_k$ of the oscillator at discrete time step $k$.
The force $u_k \in \mathcal{U} = [-1, 1]$ is used to control the system. 
The dynamics are defined as
\begin{align*}
    x_{k+1} &= x_k + \delta v_k + \frac{\delta^2u}{2} + \zeta_1,
    \\
    v_{k+1} &= v_k - K \delta v_k^3 + \delta u_k + \zeta_2,
\end{align*}
where $K = 0.0075$ is the (nonlinear) damping coefficient and $\delta = 1$ is the time discretisation step.
The stochastic noise $\zeta_k$ is a Gaussian random variable, such that $\zeta \sim \mathcal{N}([0, 0]^\top, \diag{0.25, 0.25]})$, where $\mathcal{N}(\mu,\Sigma)$ denotes a Gaussian with mean $\mu$ and covariance $\Sigma$, and the diagonal matrix $\diag{z}$ is the square matrix with $z$ on the diagonal and zero elsewhere.

\paragraph{Reach-avoid specification.}
We consider the infinite-horizon reach-avoid task to reach a state in $\xGoal = [-1, 1] \times [-1, 1]$ while avoiding states in $\xUnsafe = \RR^2 \setminus ([-10, 10] \times [-10, 10])$, which models avoiding positions and velocities above $+10$ or below $-10$.

\paragraph{Abstraction.}
We construct the abstraction based on a uniform partition with $40 \times 40$ states. 
In each state, we use 11 $\times$ 11 state samples, 11 $\times$ 11 control samples per state sample, and 11 $\times$ 11 voxels. 
We use a maximum scaling factor of $\Lambda = 3.0$ to compute the enabled actions and $N = 10\,000$ noise samples to compute the probability intervals using \theoremref{thm:PAC_interval}.
The resulting IMDP has $1\,602$ states, $25\,929$ actions, and $669\,727$ transitions. Generating the abstraction takes approximately $55$ minutes, and computing an optimal policy using PRISM takes less than $1$ second.

\begin{table}[b!]
\floatconts
{tab:results}%
{\caption{The reach-avoid probabilities from a fixed initial state $x_I$ for the dynamical systems, using either the Clopper-Pearson interval (CP) or scenario approach (Scen.) to compute intervals with either $N=1\,000$ or $N=10\,000$ samples.\vspace{-2em}}}
{%
\scriptsize\begin{tabular}{llllll}
\toprule
\bfseries System & \bfseries N=1k, CP, $\mathbf{\Lambda > 1}$ & \bfseries N=10k, CP, $\mathbf{\Lambda > 1}$ & \bfseries N=1k, Scen., $\mathbf{\Lambda > 1}$  & \bfseries N=10k, Scen., $\mathbf{\Lambda > 1}$ & \bfseries N=10k, CP, $\mathbf{\Lambda = 1}$ \\
\midrule
Pendulum & 0.343 & 0.761 & 0.203 & 0.732 & 0.000 \\
Oscillator & 0.225 & 0.471 & 0.149 & 0.437 & 0.000 \\
Car Parking & 0.016 & 0.572 & 0.001 & 0.490 & 0.224 \\
\bottomrule
\end{tabular}
}
\end{table}

\begin{figure}[b!]
\floatconts
{fig:results_pendulum}%
{\vspace{-1.5em}\caption{Reach-avoid probabilities $\satprob^\system_{\policy}(\xGoal, \xUnsafe,\horizon)$ for pendulum with (a) probability intervals from \theoremref{thm:PAC_interval} and (b) the approach from \cite{DBLP:conf/aaai/BadingsA00PS22}, both with $N = 10\,000$ samples. Fig.~(c) shows simulated trajectories under the resulting policy from \cref{eq:refined_controller} for our method.}}%
{%
\subfigure[Clopper-Pearson intervals]{%
\label{fig:results_pendulum:CP}%
\includegraphics[height=4.2cm]{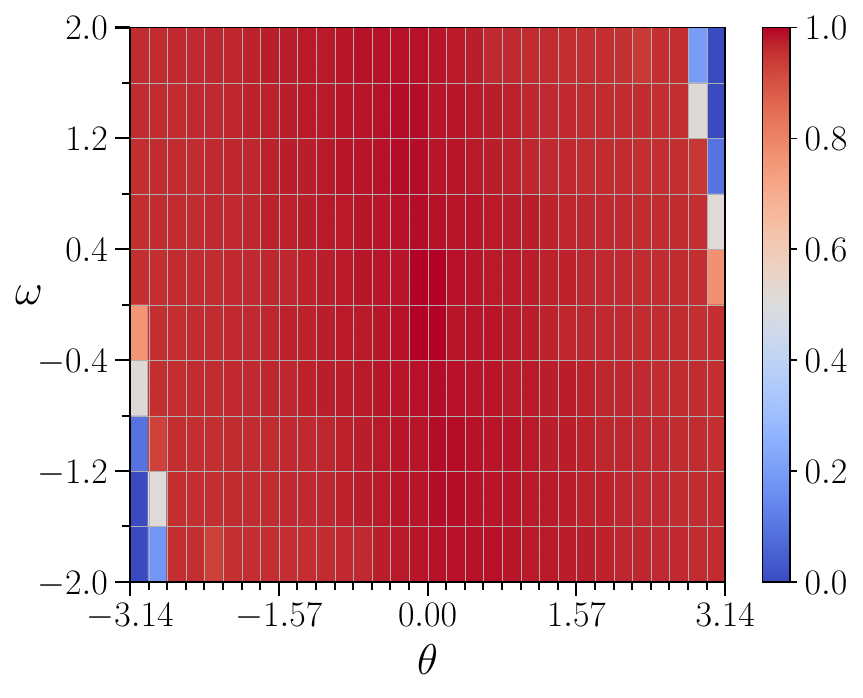}
}\hfill
\subfigure[Scenario approach intervals]{%
\label{fig:results_pendulum:scenario}%
\includegraphics[height=4.2cm]{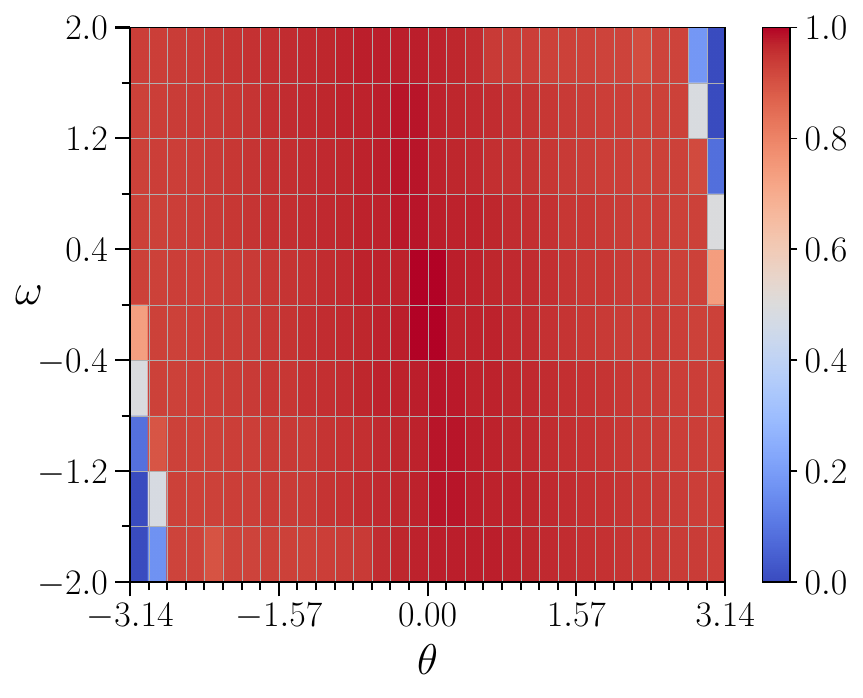}
}\hfill
\subfigure[Simulation under (a)]{%
\label{fig:results_pendulum:traj}%
\includegraphics[height=4.05cm]{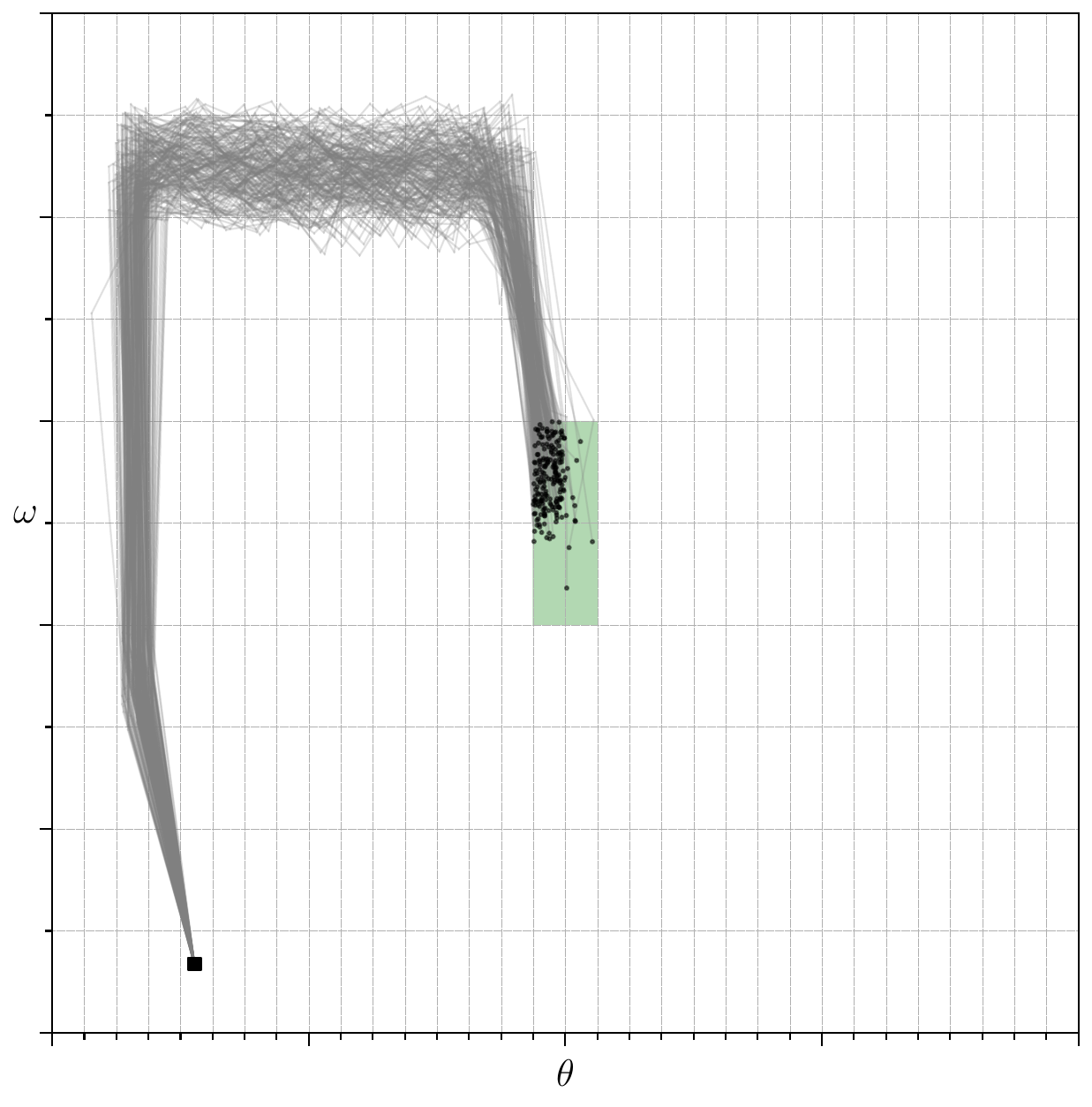}
}
}
\end{figure}

\begin{figure}[b!]
\floatconts
{fig:results_oscillator}%
{\vspace{-1.5em}\caption{Reach-avoid probabilities $\satprob^\system_{\policy}(\xGoal, \xUnsafe,\horizon)$ for the oscillator with (a) probability intervals from \theoremref{thm:PAC_interval} and (b) the approach from \cite{DBLP:conf/aaai/BadingsA00PS22}, both with $N = 10\,000$ samples. Fig.~(c) shows simulated trajectories under the resulting policy from \cref{eq:refined_controller} for our method.}}%
{%
\subfigure[Clopper-Pearson intervals]{%
\label{fig:results_oscillator:CP}%
\includegraphics[height=4.2cm]{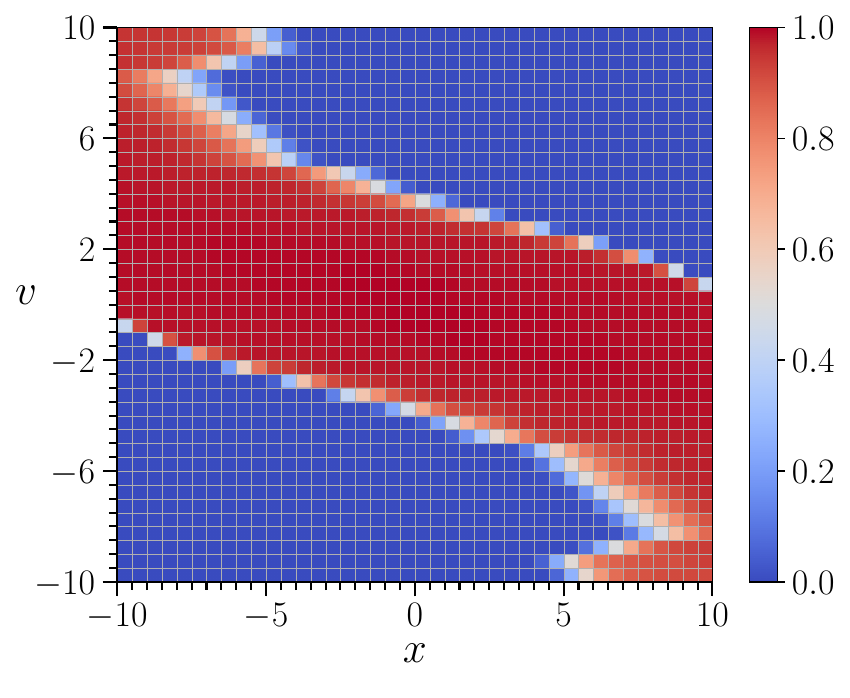}
}\hfill
\subfigure[Scenario approach intervals]{%
\label{fig:results_oscillator:scenario}%
\includegraphics[height=4.2cm]{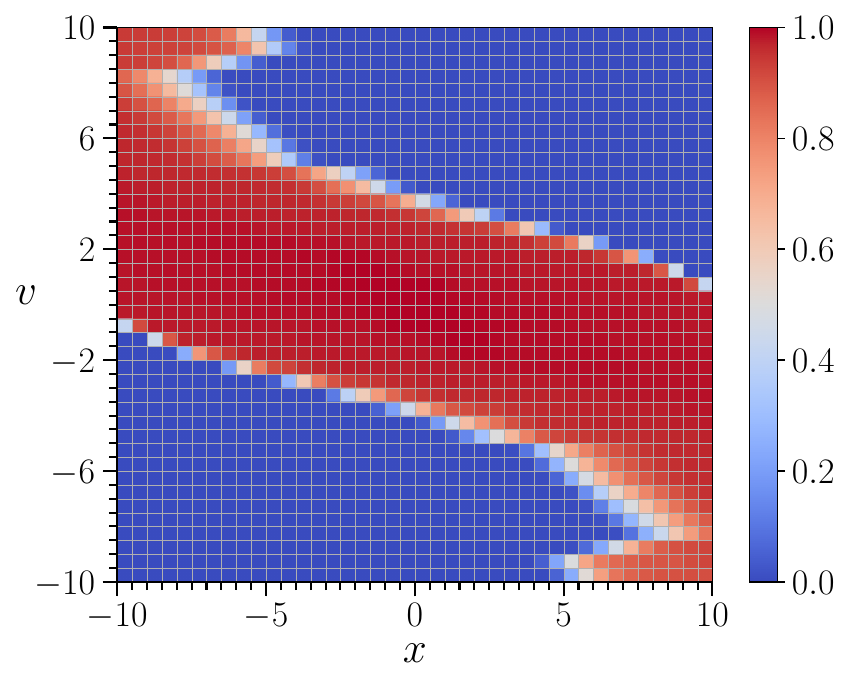}
}\hfill
\subfigure[Simulation under (a)]{%
\label{fig:results_oscillator:traj}%
\includegraphics[height=4.05cm]{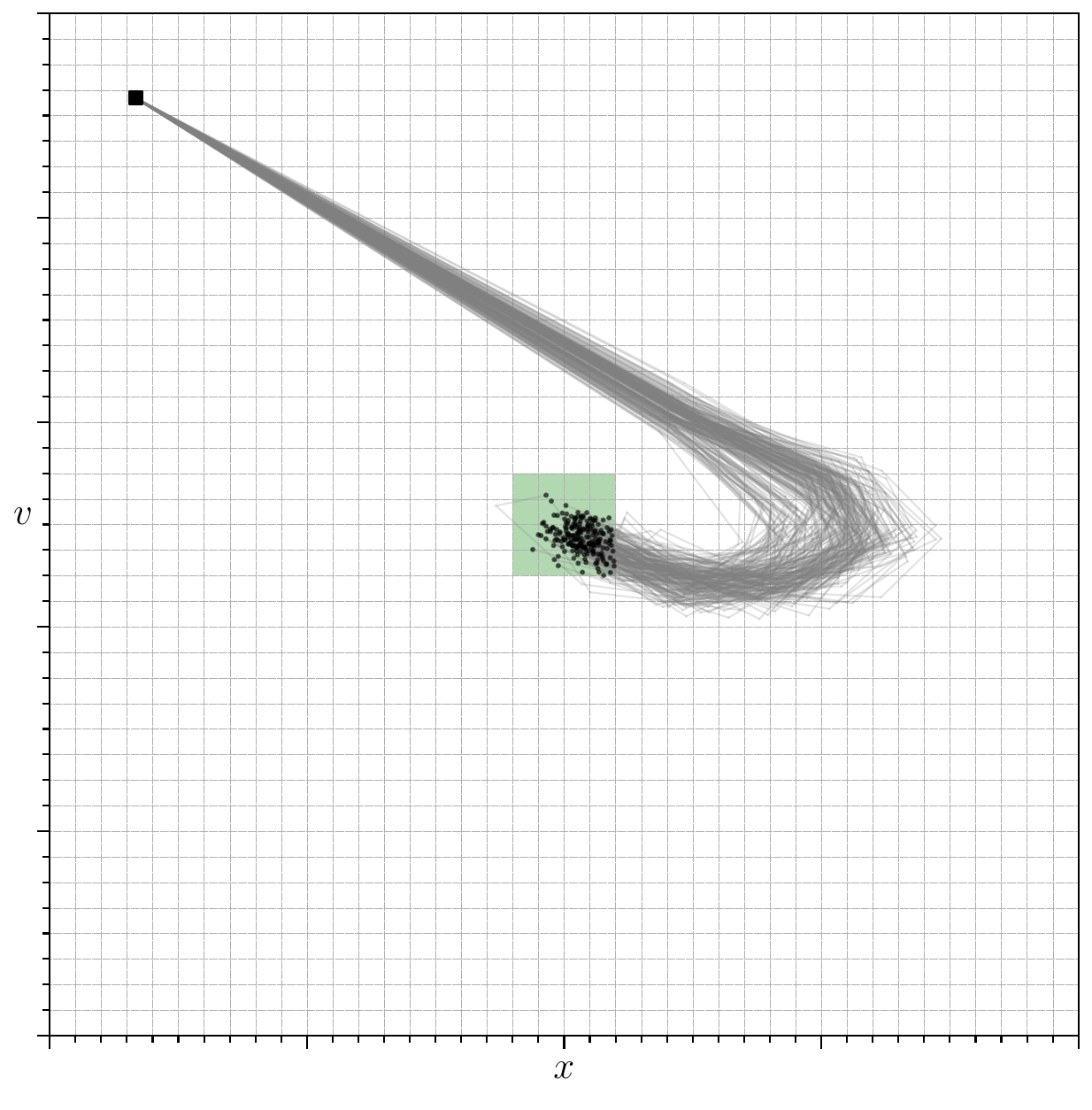}
}
}
\end{figure}

\newpage
\subsection{Experimental results}

We present all experimental results that we omitted from the main paper due to space limitations.
In \tableref{tab:results}, we present the reach-avoid probabilities across all benchmarks and cases, from a fixed initial state $x_I \in \RR^n$.
In summary, we observe higher reach-avoid probabilities for higher numbers of samples $N$ (used to compute probability intervals) and a higher $\Lambda$ (used to compute enabled actions).
Furthermore, the Clopper-Pearson interval indeed leads to tighter probability intervals than using the scenario approach.

\paragraph{Reach-avoid probabilities and trajectories.}
First, we present the heatmaps and trajectories analogous to \figureref{fig:results_car} for the other two benchmarks.
The corresponding results are presented in \figureref{fig:results_pendulum} (for the pendulum) \figureref{fig:results_oscillator} (for the harmonic oscillator).
Again, we observe that using the scenario approach leads to slightly weaker reach-avoid probabilities.

\begin{figure}[t!]
\floatconts
{fig:results_reducedLambda}%
{\vspace{-1.5em}\caption{Reach-avoid probabilities $\satprob^\system_{\policy}(\xGoal, \xUnsafe,\horizon)$ for all three benchmarks, with a reduced upper bound of $\Lambda = 1$ on each $\lambda_{i \to j}$.
We use \theoremref{thm:PAC_interval} with $N = 10\,000$ to compute probability intervals.
}}%
{%
\vspace{2em}
\subfigure[Car parking]{%
\label{fig:results_reducedLambda:car}%
\includegraphics[height=4cm]{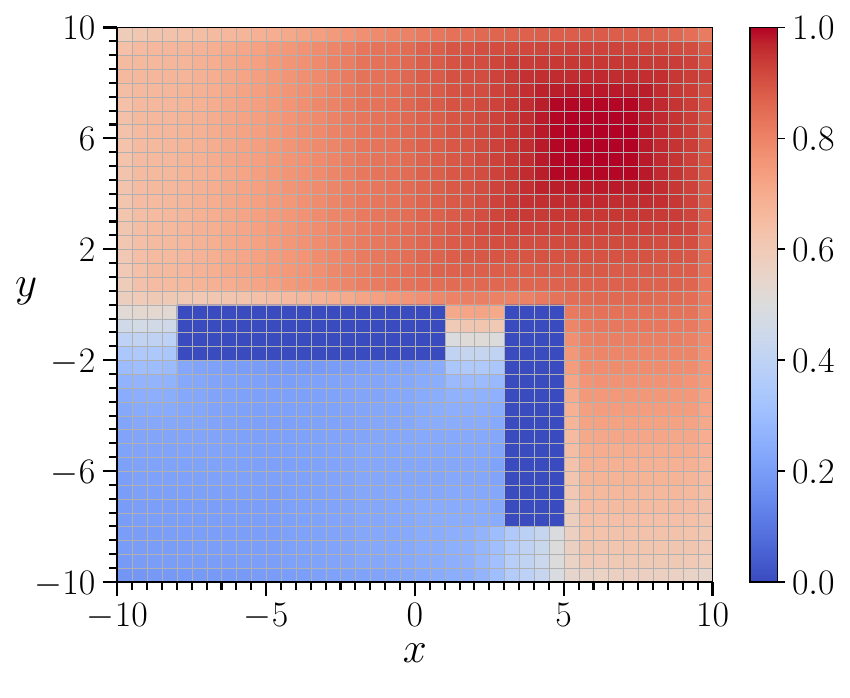}
}\hfill
\subfigure[Pendulum]{%
\label{fig:results_reducedLambda:pendulum}%
\includegraphics[height=4cm]{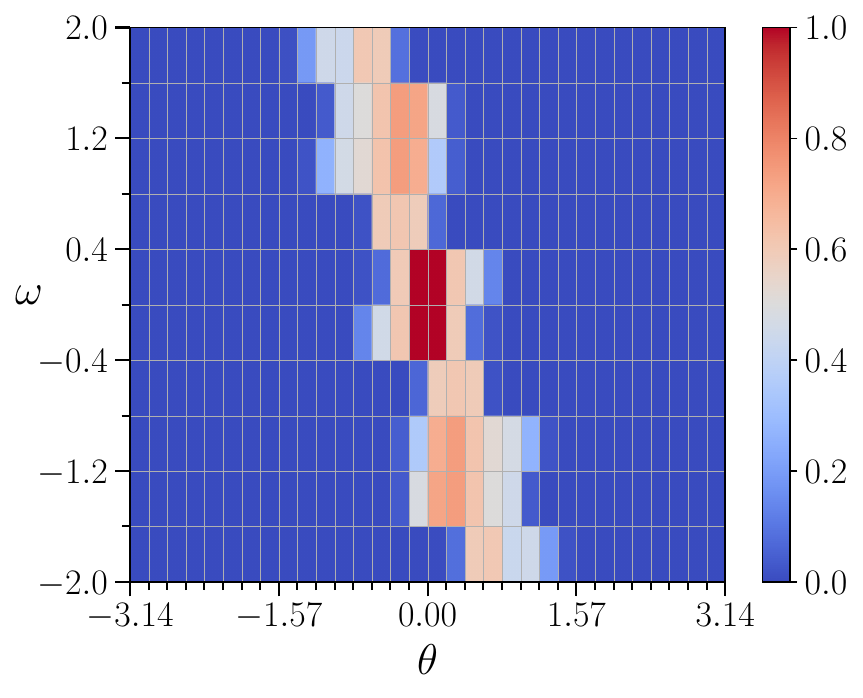}
}\hfill
\subfigure[Harmonic oscillator]{%
\label{fig:results_reducedLambda:oscillator}%
\includegraphics[height=4cm]{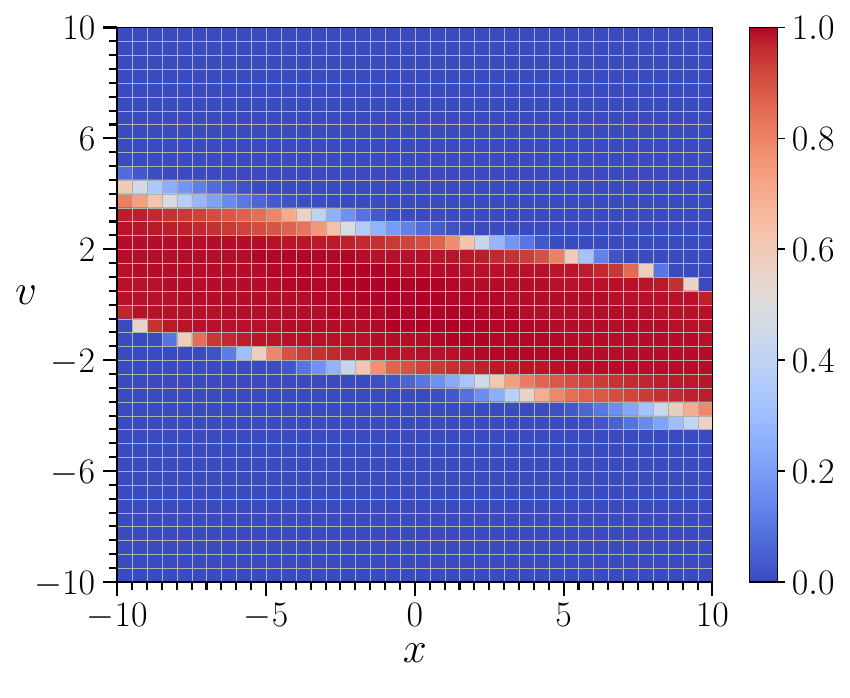}
}
}
\end{figure}

\paragraph{Reduced upper bound $\Lambda$.}
We investigate the effect of using a lower value for the hyperparameter $\Lambda$.
Recall that $\Lambda$ is the maximum allowed value of $\lambda_{i \to j}$ and thus controls the size of the backward reachable set of each IMDP action.
The heatmaps of the result reach-avoid probabilities for all three benchmarks are presented in \figureref{fig:results_reducedLambda}.
By comparing these results with those presented earlier for higher values of $\Lambda$, we observe that reducing $\Lambda$ reduces the reach-avoid probabilities that we obtain significantly.

\paragraph{Lower number of samples $N$.}
Finally, we investigate the effect of using a lower number of $N = 1\,000$ samples to compute the probability intervals. The corresponding heatmaps of the reach-avoid probability are shown in \figureref{fig:results_car_lowN,fig:results_pendulum_lowN,fig:results_oscillator_lowN}.
We observe that reducing the number of samples widens the resulting probability intervals significantly, resulting in lower reach-avoid probabilities.
In particular, when the number of samples is lower, the probability of reaching the (terminal) sink state is higher, which accumulates over each state transition of the system.

\begin{figure}[t!]
\floatconts
{fig:results_car_lowN}%
{\vspace{-1.5em}\caption{Reach-avoid probabilities $\satprob^\system_{\policy}(\xGoal, \xUnsafe,\horizon)$ for car parking with (a) probability intervals from \theoremref{thm:PAC_interval} and (b) the approach from \cite{DBLP:conf/aaai/BadingsA00PS22}, both with $N = 1\,000$ samples.}}%
{%
\subfigure[Clopper-Pearson intervals]{%
\label{fig:results_car_lowN:CP}%
\includegraphics[height=3.8cm]{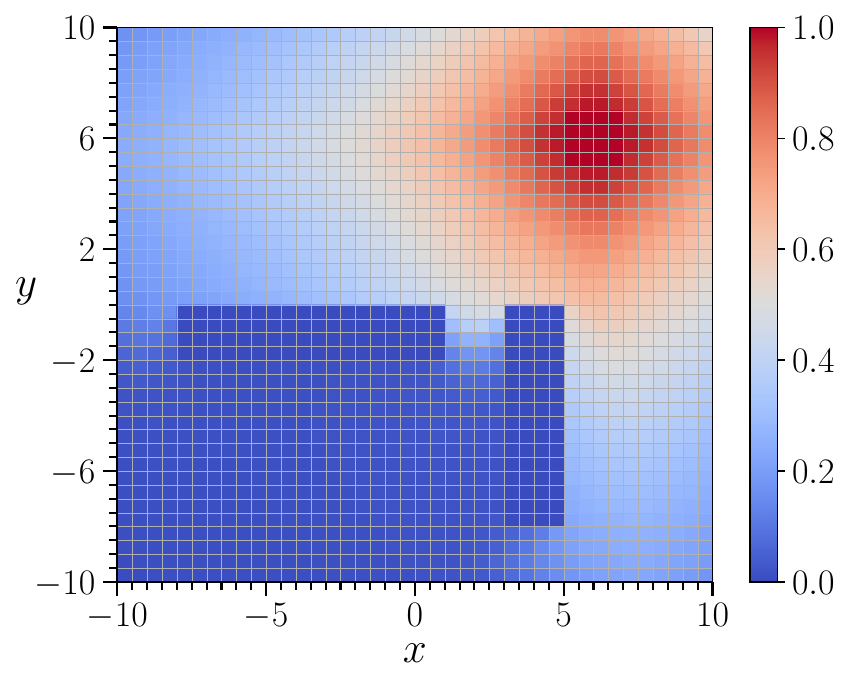}
}
\subfigure[Scenario approach intervals]{%
\label{fig:results_car_lowN:scenario}%
\includegraphics[height=3.8cm]{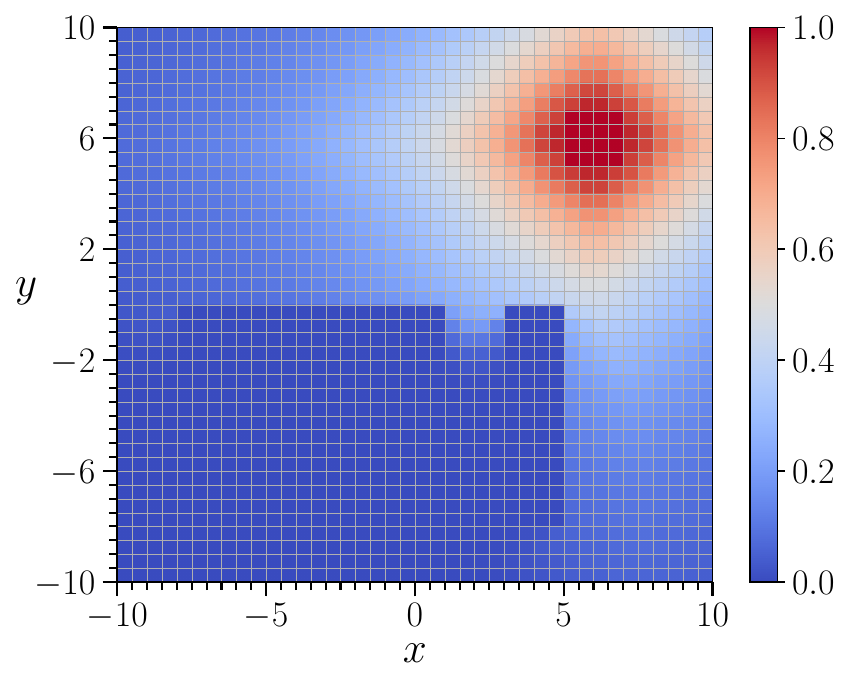}
}\hfill
}
\vspace{0.5em}
\floatconts
{fig:results_pendulum_lowN}%
{\vspace{-1.5em}\caption{Reach-avoid probabilities $\satprob^\system_{\policy}(\xGoal, \xUnsafe,\horizon)$ for pendulum with (a) probability intervals from \theoremref{thm:PAC_interval} and (b) the approach from \cite{DBLP:conf/aaai/BadingsA00PS22}, both with $N = 1\,000$ samples.}}%
{%
\subfigure[Clopper-Pearson intervals]{%
\label{fig:results_pendulum_lowN:CP}%
\includegraphics[height=3.8cm]{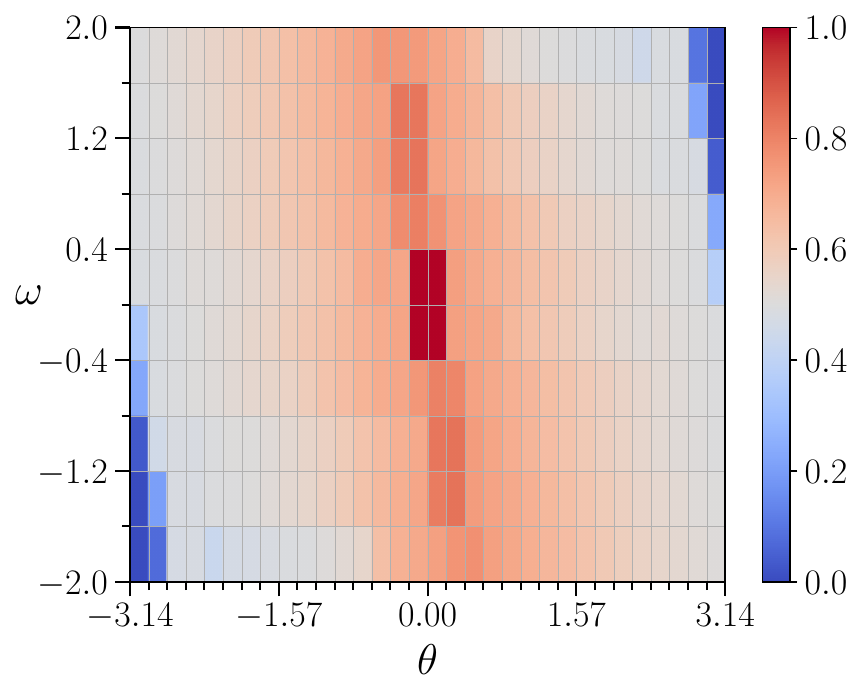}
}
\subfigure[Scenario approach intervals]{%
\label{fig:results_pendulum_lowN:scenario}%
\includegraphics[height=3.8cm]{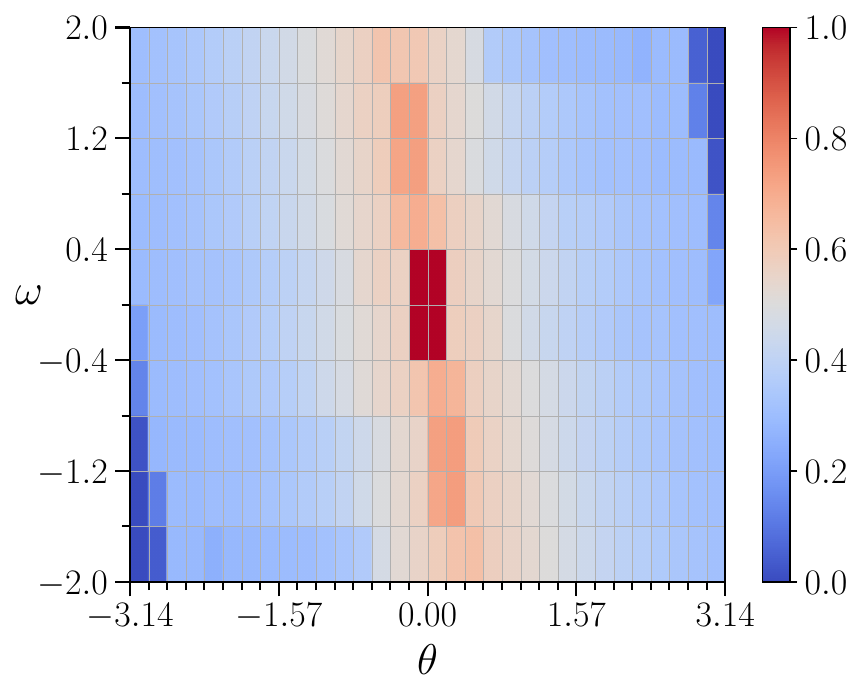}
}\hfill
}
\vspace{0.5em}
\floatconts
{fig:results_oscillator_lowN}%
{\vspace{-1.5em}\caption{Reach-avoid probabilities $\satprob^\system_{\policy}(\xGoal, \xUnsafe,\horizon)$ for the oscillator with (a) probability intervals from \theoremref{thm:PAC_interval} and (b) the approach from \cite{DBLP:conf/aaai/BadingsA00PS22}, both with $N = 1\,000$ samples.}}%
{%
\subfigure[Clopper-Pearson intervals]{%
\label{fig:results_oscillator_lowN:CP}%
\includegraphics[height=3.8cm]{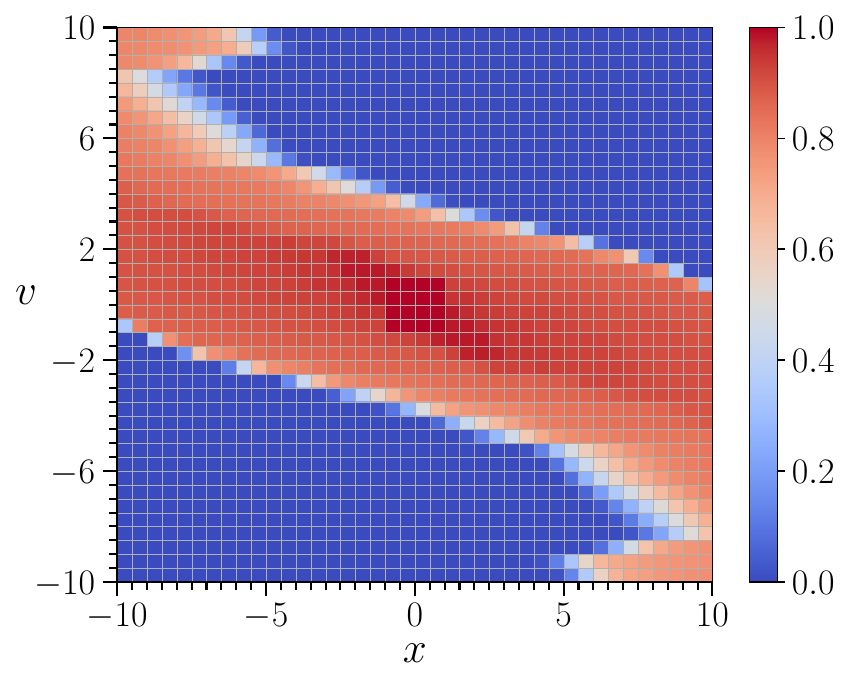}
}
\subfigure[Scenario approach intervals]{%
\label{fig:results_oscillator_lowN:scenario}%
\includegraphics[height=3.8cm]{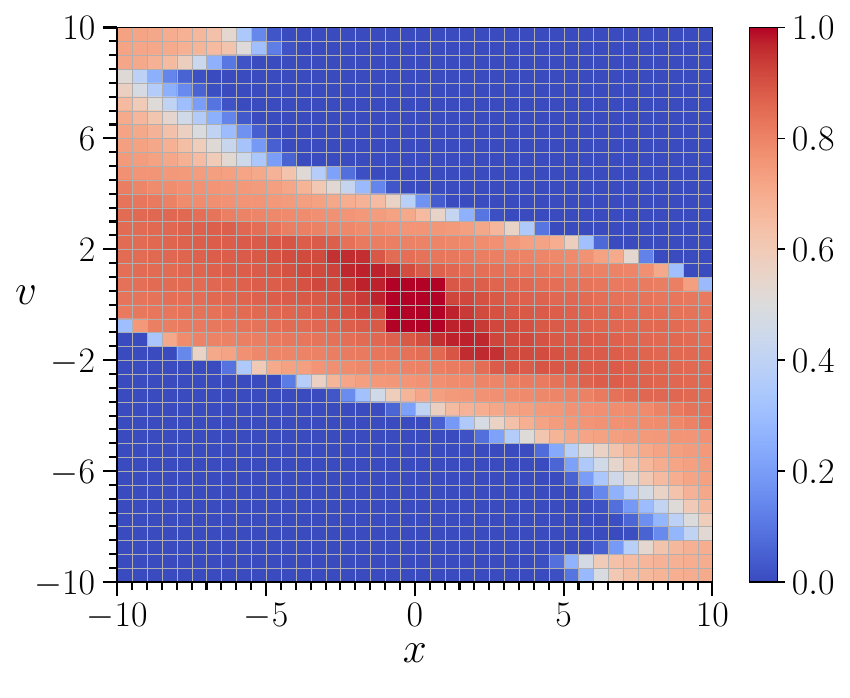}
}\hfill
}
\end{figure}%
\fi%
\end{document}